\documentclass[lettersize,journal]{IEEEtran}

\usepackage{amsmath,amssymb,amsfonts}
\usepackage{amsthm}
\usepackage{placeins}
\usepackage{caption}
\usepackage{array}
\usepackage[caption=false,font=normalsize,labelfont=sf,textfont=sf]{subfig}
\usepackage{textcomp}
\usepackage{stfloats}
\usepackage{url}
\usepackage{graphicx}
\usepackage{cite}
\usepackage{multirow}
\usepackage{booktabs}
\usepackage{float}
\usepackage{bm}
\usepackage{tikz}
\usepackage{makecell}
\usepackage[pass]{geometry}
\usetikzlibrary{arrows.meta,positioning,calc}

\newtheorem{theorem}{Theorem}
\newtheorem{lemma}{Lemma}
\newtheorem{proposition}{Proposition}
\newtheorem{corollary}{Corollary}
\theoremstyle{definition}
\newtheorem{remark}{Remark}

\newcommand{\diag}{\mathrm{diag}}
\newcommand{\E}{\mathbb{E}}
\renewcommand{\jmath}{\mathrm{j}}

\begin{document}
\title{Gain Before Tracking: Compensating Power-Amplifier Nonlinearity and Phase Noise in AFDM Uplinks}

\author{Mohammad~Tasin and Sadman~Sakib~Rahman%
\thanks{M.~Tasin and S.~S.~Rahman are with the Department of Electrical and
Electronic Engineering, Bangladesh University of Engineering and Technology
(BUET), Dhaka 1205, Bangladesh (e-mail: tasin2006072@gmail.com;
ssrahman@eee.buet.ac.bd).}%
}

\markboth{Preprint, October 2026}%
{Tasin and Rahman: Gain Before Tracking for AFDM Uplinks}

\maketitle

%% ---------------------------------------------------------------
\begin{abstract}
Affine frequency division multiplexing (AFDM) is a candidate waveform for fast-moving vehicles and satellites, but real transmitters add power-amplifier (PA) compression and oscillator phase noise, and no AFDM receiver corrects both. We show that AFDM's transform treats every component equally in magnitude, so the PA's effect reduces to one gain shared by all components and phase noise to a common rotation that OFDM's standard tracker removes. The order matters: the tracker relies on symbol decisions that PA compression corrupts, so the gain should be corrected first. On a six-path, fractional-Doppler channel with 256 subcarriers and a standard-compliant PA back-off, doing so lowers the uncoded bit-error rate by 7.5--12\% for AFDM and about 5\% for OFDM, and saves 0.5~dB at 10\% block error with a rate-3/4 code; the effect is not significant at light back-off. A closed-form gain matches a blind estimate and, without pilots, performs comparably to a pilot-based joint gain estimate that costs 6\% overhead. Controls show that AFDM's edge over OFDM is mostly spreading across the channel, not Doppler robustness: single-carrier DFT-s-OFDM matches it at equal back-off. 
\end{abstract}

\begin{IEEEkeywords}
AFDM, Bussgang decomposition, doubly-dispersive channels, phase noise, power amplifier nonlinearity, vehicular communications.
\end{IEEEkeywords}

%% ---------------------------------------------------------------
\section{Introduction}
\label{sec:intro}
\IEEEPARstart{F}{ast-moving} users---cars, trains, drones, and low-earth-orbit satellites---see a radio channel that changes quickly in both time and frequency (a ``doubly-dispersive'' channel), which is one reason future 6G systems look beyond today's OFDM waveform~\cite{dang2020what6g,afdm6g2026}. Affine frequency division multiplexing~(AFDM) is a leading candidate. It replaces OFDM's Fourier transform with the discrete affine Fourier transform~(DAFT)---multiply by a chirp, apply an FFT, multiply by a second chirp---so that every symbol is spread across the channel's delay--Doppler diversity at essentially OFDM's receiver cost~\cite{bemani2021afdm,bemani2023afdm}.

That promise assumes ideal hardware. Real uplink transmitters drive a power amplifier near saturation, and their oscillators add phase noise. Existing AFDM hardware studies mostly measure the damage~\cite{sui2026hwi,ctafdm2026,iqi2026afdm}; the few fixes address one impairment at a time, namely phase noise~\cite{pnafdm2026} and receiver IQ imbalance~\cite{iqi2025afdm}, and to our knowledge none corrects amplifier nonlinearity and phase noise together. OTFS, a related waveform, has only a black-box neural receiver for two narrow imperfections~\cite{dl_otfs_hwi2023}.

Our starting point is a structural fact: the DAFT matrix, like the DFT matrix, is \emph{equimodular}---every entry has the same magnitude $1/\sqrt{N}$. Consequently the PA's Bussgang gain and distortion are approximately the same on every AFDM component, and phase noise acts as one common rotation plus leakage, so OFDM's standard decision-directed tracker~\cite{pollet1995,petrovic2005} carries over. The less obvious fact is the \emph{order} of the two corrections: the tracker trusts its own symbol decisions, and PA compression shrinks the constellation toward the decision thresholds and corrupts them, so the gain must be corrected before tracking.

\subsection{Related Work}
\label{sec:related}
The closest prior work is the hardware-impairment analysis of MIMO-AFDM in~\cite{sui2026hwi}, which derives a Bussgang-based LMMSE receiver with a real gain and shows AFDM keeps its diversity advantage over OFDM under impairments; it has no phase tracker and does not study the order of corrections, and we compare against it in Section~\ref{sec:gain_source}. Our phase-noise results build on classical OFDM analysis, namely the CPE as the frame average of $e^{j\varphi(n)}$ and the Wiener-ICI power law~\cite{pollet1995,armada2001,petrovic2005}; Theorem~\ref{thm:kernel}(a)(b) restate these for the affine--Fourier family, and the converse~(c), to our knowledge, is new but plays no role in the receiver. For the PA we use the Bussgang decomposition~\cite{bussgang1952,demir2021distortion}, and for the blind gain a moment-based estimator from the SNR-estimation literature~\cite{pauluzzi2000}. We target the uplink because the cost-constrained PA sits in the mobile terminal; digital predistortion is possible there~\cite{dpd_survey2010}, but it adds a feedback path and calibration at the terminal, so a receiver-side correction is attractive, and we leave a comparison for future work. Placing a scalar gain ahead of a decision-directed phase tracker is ordinary receiver practice, and 5G NR phase-tracking pilots estimate phase, or a complex gain, from known symbols; we do not claim the principle is new for OFDM, and we have not surveyed OFDM joint nonlinearity and phase-noise compensation exhaustively. What we add is a quantified order effect for AFDM under PA compression, a closed-form value for the scalar, and a comparison with a pilot-based baseline (Section~\ref{sec:gain_source}).

The first version of this idea was checked in a small toy setup (64 subcarriers, two paths, integer Doppler). Here we test it in a more demanding regime: 256--512 subcarriers, a six-path channel with fractional Doppler, a PA back-off whose error-vector magnitude (EVM) satisfies the 5G NR limit~\cite{3gpp38101}, channel sounding with thermal noise and compressed pilots, and convolutional coding. We make four contributions.

\textbf{(1) Structure.} We show that equimodularity is what lets OFDM's phase-noise facts transfer (Theorem~\ref{thm:kernel}) and that, under a Gaussian approximation, PA distortion has the same statistics on every AFDM component, so one scalar suffices (Proposition~\ref{prop:uniform}). These results motivate the receiver; they do not predict the effect sizes below.

\textbf{(2) Gain before tracking.} With controls, applying the gain before the tracker lowers uncoded BER by 7.5--12\% for AFDM and about 5\% for OFDM relative to the reverse order in the realistic regime (lower BER on 25--30 of 30 channels), and BLER by 9--28\% at rate 3/4 and 3 or 5~dB back-off, which is 0.5~dB at 10\% BLER at the NR-compliant back-off and 1.6~dB at 3~dB. The effect needs compression: with a linear PA a data-derived gain applied first is no better, and at $8$~dB back-off the effect is not significant. Gain-first gives 8--10\% fewer errors in the tracker's first decisions and a phase estimate about 20\% closer to the oracle (Section~\ref{sec:order_results}).

\textbf{(3) Three ways to get the gain.} The closed-form Bussgang gain, a blind moment estimator, and a gain calibrated on a known-symbol preamble of 4096 symbols. The blind estimator ties the closed form. The closed form, which uses no pilots, ties a PT-RS-style joint complex gain estimated from 16 in-frame pilots (6.25\% overhead), so what matters is getting the amplitude right, not how. The preamble gain is better by a further 6--14\% in our idealized sounding.

\textbf{(4) A controlled AFDM-versus-OFDM comparison.} With identical hardware, AFDM has 14--31\% lower uncoded BER than OFDM over $\nu_{\max}=0.1$--$4$ (not significant at $N=64$, $6.1\%$ at $3$~dB back-off) and needs 2--3~dB less SNR at coding rate 3/4. But OCDM and DFT-s-OFDM obtain much of the same edge, and DFT-s-OFDM matches AFDM at equal back-off, so the edge is mostly spreading (delay diversity), not Doppler robustness or chirp-specific, and with a rate-1/2 code it nearly vanishes.

Section~\ref{sec:system_model} gives the model, Section~\ref{sec:receiver} the receiver, Section~\ref{sec:theory} the theory, Section~\ref{sec:experiments} the results, and Section~\ref{sec:conclusion} the limits. Proofs, experimental details and additional diagnostics are in the appendices.

%% ---------------------------------------------------------------
\section{System Model}
\label{sec:system_model}
\subsection{DAFT Signal Model}
\label{sec:simo}
Let $\boldsymbol{\Lambda}_{c}=\diag\{e^{-j2\pi c n^2}\}_{n=0}^{N-1}$ and let $\mathbf{F}$ be the unitary DFT. The DAFT demodulation matrix is
\begin{equation}
\mathbf{A}=\boldsymbol{\Lambda}_{c_2}\mathbf{F}\boldsymbol{\Lambda}_{c_1},\qquad
c_1=\frac{2(\lceil\nu_{\max}\rceil+2)+1}{2N},\;\; c_2=\frac{1}{2N},
\label{eq:daft}
\end{equation}
where $\nu_{\max}$ is the maximum Doppler in subcarrier spacings; the choice of $c_1$ gives margin beyond the integer Doppler~\cite{bemani2023afdm}. Setting $c_1=c_2=0$ gives OFDM, which we use as the reference waveform throughout. Both matrices are unitary with all entries of modulus $1/\sqrt{N}$ (\emph{equimodular}). The transmit frame is $\mathbf{s}=\mathbf{A}^H\mathbf{x}$ for $N$ Gray-mapped 16-QAM symbols $\mathbf{x}$, preceded by a chirp-periodic prefix (CPP) of $L=16$ samples (an ordinary cyclic prefix for OFDM). The subcarrier spacing is $30$~kHz, so the sample rate is $f_s=30N$~kHz.

\subsection{Doubly-Dispersive Channel}
\label{sec:channel}
The channel has $P$ paths with complex gains $h_i$, integer delays $l_i$ and Doppler shifts $\nu_i$ (in subcarrier spacings, generally fractional):
\begin{equation}
r[n]=\sum_{i=1}^{P}h_i\,s'[n-l_i]\,e^{j2\pi\nu_i n/N}+w[n],
\label{eq:chan}
\end{equation}
with $w[n]\sim\mathcal{CN}(0,\sigma_n^2)$. After CPP removal and DAFT demodulation, the effective channel is $\mathbf{H}_{\mathrm{eff}}=\mathbf{A}\mathbf{H}_t\mathbf{A}^H$ with $\mathbf{H}_t$ the time-domain channel matrix; for integer Doppler it has $P$ nonzeros per row, and fractional Doppler spreads each over neighboring bins~\cite{bemani2023afdm}. Our experiments use $P=6$, delays $\{0,1,2,4,6,9\}$ samples, an exponential power profile, Rayleigh gains normalized to unit total power, and per-path Doppler $\nu_i=\nu_{\max}\cos\theta_i$ with $\theta_i$ uniform, so every path has a fractional Doppler. We set $\nu_{\max}=1$ by default and sweep $0.1$--$4$. For orientation, $\nu_{\max}=0.1$ corresponds to $116$~km/h at $28$~GHz; $\nu_{\max}\ge1$ is a satellite or extreme-mobility regime, and a 120~km/h, 3.5~GHz link has $\nu_{\max}=0.013$.

\subsection{Power Amplifier and Phase Noise}
\label{sec:pa_model}
The PA is a memoryless Rapp nonlinearity~\cite{rapp1991} with AM/AM and AM/PM responses
\begin{equation}
g(r)=\frac{r}{[1+(r/P_{\mathrm{sat}})^{2p}]^{1/(2p)}},\qquad
\phi(r)=\frac{\alpha r^2}{P_{\mathrm{sat}}^2+r^2},
\label{eq:rapp}
\end{equation}
so that $\tilde s_n=g(|s_n|)e^{j\angle s_n+j\phi(|s_n|)}$, with $p=2$ and $\alpha=0.1$~rad. The operating point is the input back-off $\mathrm{IBO}=20\log_{10}(P_{\mathrm{sat}}/\sigma_s)$. Our default is $\mathrm{IBO}=5$~dB, where the PA's EVM is $9.1\%$, inside the $12.5\%$ limit for 16-QAM uplinks~\cite{3gpp38101}; at $3$~dB the EVM is $14.0\%$ and violates that limit, which we keep as a stress case, and $8$~dB ($3.9\%$) is a lightly driven PA. The PA is applied at the symbol rate, without oversampling.

Oscillator phase noise is a Wiener process, $\varphi(n)=\varphi(n-1)+w_\varphi(n)$ with $w_\varphi\sim\mathcal{N}(0,\sigma_w^2)$, $\sigma_w^2=2\pi\Delta f/f_s$~\cite{mathecken2011,petrovic2005}. The linewidth is $\Delta f=100$~Hz, below the threshold $\Delta f^*$ of Section~\ref{sec:dfstar} at the default back-off. The transmitted frame is $\mathbf{s}'=\boldsymbol{\Phi}\,f_{\mathrm{PA}}(\mathbf{s})$ with $\boldsymbol{\Phi}=\diag\{e^{j\varphi(n)}\}$, independent per frame, and the DAFT-domain received signal is
\begin{equation}
\mathbf{y}=\mathbf{H}_{\mathrm{eff}}\,\mathbf{A}\boldsymbol{\Phi}\,f_{\mathrm{PA}}\!\bigl(\mathbf{A}^H\mathbf{x}\bigr)+\mathbf{n}.
\label{eq:rx_daft}
\end{equation}

\subsection{Channel Sounding and Baseline Receiver}
\label{sec:baseline}
We estimate $\mathbf{H}_{\mathrm{eff}}$ by sounding: for each chirp-subcarrier $m$, a pilot frame carrying the single symbol $\sqrt{N}\,\mathbf{e}_m$ goes through the full transmit chain, so the time-domain pilot has the data's RMS level and the PA compresses it. Column $m$ of $\hat{\mathbf{H}}_{\mathrm{eff}}$ is the DAFT-domain response, with independent thermal noise per frame. The $N$ sounding frames share one phase-noise realization, an idealization (phase noise slow relative to the preamble) that leaves a static phase offset of a few degrees in $\hat{\mathbf{H}}_{\mathrm{eff}}$. The preamble overhead is not charged. The baseline LMMSE equalizer is
\begin{equation}
\hat{\mathbf{s}}=\bigl(\hat{\mathbf{H}}_{\mathrm{eff}}^H\hat{\mathbf{H}}_{\mathrm{eff}}+\sigma_n^2\mathbf{I}\bigr)^{-1}\hat{\mathbf{H}}_{\mathrm{eff}}^H\mathbf{y}.
\label{eq:lmmse}
\end{equation}
We solve it densely at $O(N^3)$; the sparsity of $\mathbf{H}_{\mathrm{eff}}$ at integer Doppler admits cheaper solvers~\cite{bemani2023afdm}, which we do not use. Every receiver below post-processes $\hat{\mathbf{s}}$.

%% ---------------------------------------------------------------
\section{Receiver}
\label{sec:receiver}
The receiver has two stages after the LMMSE equalizer (Fig.~\ref{fig:receiver}), neither trained: a scalar gain correction (Stage~1) and a decision-directed phase tracker (Stage~2).

\begin{figure}[!t]
\centering
\resizebox{\columnwidth}{!}{%
\begin{tikzpicture}[>=Stealth,node distance=4.2mm,
 blk/.style={draw,rounded corners=2pt,minimum height=7mm,inner sep=2.5pt,font=\footnotesize,align=center}]
\node[blk] (y) {$\mathbf{y}$};
\node[blk,right=of y] (eq) {LMMSE\\$\hat{\mathbf{H}}_{\mathrm{eff}}$};
\node[blk,right=of eq,fill=orange!15] (g) {Stage 1\\gain $\hat g$};
\node[blk,right=of g,fill=green!12] (t) {Stage 2\\CPE tracker};
\node[blk,right=of t] (d) {demap /\\decode};
\draw[->] (y)--(eq); \draw[->] (eq)--(g); \draw[->] (g)--(t); \draw[->] (t)--(d);
\node[font=\scriptsize,align=center,below=3mm of g] (src) {closed form \eqref{eq:wiener}\\blind M2M4\\preamble \eqref{eq:cal}};
\draw[->,dashed] (src)--(g);
\end{tikzpicture}}
\caption{Proposed receiver. The gain is applied before the tracker; it can come from the PA model, from blind moments, or from a known-symbol preamble.}
\label{fig:receiver}
\end{figure}
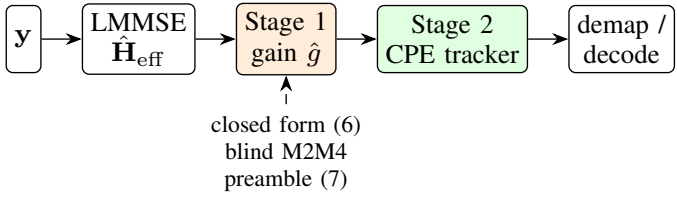

\subsection{Stage 1: Scalar Gain Correction}
\label{sec:bussgang_layer}
For a zero-mean circular Gaussian input $s$ the Bussgang theorem~\cite{bussgang1952,demir2021distortion} writes the PA output as $f_{\mathrm{PA}}(s)=Bs+\eta$ with $\E\{s^*\eta\}=0$, where the complex gain $B$ and distortion power $\sigma_\eta^2$ follow from the IBO (Section~\ref{sec:bussgang_const}). By Proposition~\ref{prop:uniform}, $B$ and $\sigma_\eta^2$ are approximately the same on every DAFT component. Under this model the linear scalar that best recovers the pre-PA signal from its PA output is
\begin{equation}
W=\frac{B^*}{|B|^2+\sigma_\eta^2}.
\label{eq:wiener}
\end{equation}
(It is the per-component LMMSE scalar under the Bussgang model, not the MMSE detector for the whole chain.) Stage~1 outputs $\mathbf{v}=\hat g\,\hat{\mathbf{s}}$, and $\hat g$ can be obtained three ways:

\emph{Closed form:} $\hat g=W$, computed once from the IBO and the PA model.

\emph{Blind moments:} with $M_2=\E|\hat s|^2$ and $M_4=\E|\hat s|^4$ over unlabeled symbols and the unit-power 16-QAM kurtosis $k_a=1.32$, the M2M4 estimator~\cite{pauluzzi2000} gives the amplitude $|\hat g|^{-1}=[(2M_2^2-M_4)/(2-k_a)]^{1/4}$; it needs no PA model and no pilots.

\emph{Preamble calibration:} from $K$ known-symbol frames, the regression coefficient
\begin{equation}
\hat g^{-1}=\sum_k \hat s_k x_k^*\Big/\sum_k |x_k|^2
\label{eq:cal}
\end{equation}
is a complex gain, so it also removes the static phase offset left in $\hat{\mathbf{H}}_{\mathrm{eff}}$ by sounding. We use $K=16$ frames at $N=256$ (4096 symbols).

\subsection{Stage 2: Decision-Directed CPE Tracker}
\label{sec:cpe_tracker}
Phase noise creates a common phase error $J_0=\frac1N\sum_n e^{j\varphi(n)}$ on all subcarriers (Corollary~\ref{cor:tracker}). The tracker estimates it from hard decisions $c_k$ (nearest 16-QAM points) as
\begin{equation}
\hat\theta=\angle\Bigl(\sum_k v_k c_k^*\Bigr),\qquad v_k^{\mathrm{tr}}=v_k e^{-j\hat\theta},
\label{eq:cpe}
\end{equation}
iterated three times with refreshed decisions. A tracker followed by the gain is the \emph{tracker-first} (post) order; Stage~1 followed by the tracker is \emph{gain-first} (pre).

\subsection{Why the Gain Must Come First}
\label{sec:order}
The angle in~\eqref{eq:cpe} ignores a common scale only if the decisions $c_k$ are right, and compression makes them wrong. At $\mathrm{IBO}=5$~dB, $|B|=0.914$ moves the outer 16-QAM level from $3/\sqrt{10}$ to $0.867$, shrinking its margin to the decision threshold $2/\sqrt{10}$ from $0.316$ to $0.235$ ($26\%$). Distortion, leakage and noise then push more symbols across thresholds, and each wrong decision biases $\hat\theta$. Restoring the gain afterwards recovers the scale but not the phase error already committed. The mechanism predicts that the order effect should vanish for a linear PA and fade as the back-off grows; Section~\ref{sec:order_results} tests both.

\subsection{Cost}
\label{sec:cost}
Stage~1 is one complex multiplication per subcarrier and Stage~2 three passes of a hard decision and an $N$-term sum, so both are $O(N)$; the dense LMMSE equalizer of~\eqref{eq:lmmse} dominates the cost at $O(N^3)$ and is common to all receivers. The closed-form gain needs the PA model and IBO, the M2M4 gain needs neither but one block of unlabeled symbols, and the preamble gain needs $K$ known-symbol frames.

%% ---------------------------------------------------------------
\section{Theory}
\label{sec:theory}
We establish what the structure gives us: that the Bussgang model applies equally to AFDM (Lemma~\ref{lem:equimod}), that one scalar describes all components (Proposition~\ref{prop:uniform}), and that OFDM's phase-noise facts transfer (Theorem~\ref{thm:kernel}). The theory does not predict the order effect quantitatively or the AFDM-versus-OFDM BER difference; those are measured in Section~\ref{sec:experiments}.

\subsection{Equimodularity and the Gaussian Approximation}
\label{sec:gaussianity}
\begin{lemma}[Equimodularity of the DAFT basis]
\label{lem:equimod}
Every entry of $\mathbf{A}^H$ in~\eqref{eq:daft} has modulus $1/\sqrt{N}$. Hence each time-domain sample of an AFDM frame is a sum of $N$ equal-weight terms with i.i.d.\ symbols, so by the central limit theorem its distribution is approximately complex Gaussian for large $N$, as for OFDM (Berry--Esseen bounds the error of the distribution function at $O(N^{-1/2})$). The PA's Bussgang constants $(B,\sigma_\eta^2)$ depend only on this marginal and are therefore approximately equal to OFDM's; we treat this as an approximation and check it numerically in the appendices.
\end{lemma}
\begin{IEEEproof}[Proof sketch]
$\boldsymbol{\Lambda}_{c_1},\boldsymbol{\Lambda}_{c_2}$ are diagonal and unitary, so $|[\mathbf{A}^H]_{nk}|=|[\mathbf{F}^H]_{nk}|=1/\sqrt{N}$. The rest is the Berry--Esseen bound for weighted sums; see Appendix~\ref{sec:supp_lemma1}.
\end{IEEEproof}

\subsection{Bussgang Constants for the Rapp PA}
\label{sec:bussgang_const}
Let $\kappa=10^{-\mathrm{IBO}/20}$ and $u=|s|^2\sim\mathrm{Exp}(1)$ for a unit-power Gaussian input driven at $\kappa s$ and scaled back by $1/\kappa$:
\begin{equation}
\begin{aligned}
B&=\tfrac{1}{\kappa}\E_u\!\bigl[g(\kappa\sqrt{u})e^{j\phi(\kappa\sqrt{u})}\sqrt{u}\bigr],\\
\sigma_\eta^2&=\tfrac{1}{\kappa^2}\E_u\!\bigl[g^2(\kappa\sqrt{u})\bigr]-|B|^2 .
\end{aligned}
\label{eq:bussgang_afdm}
\end{equation}
At $\mathrm{IBO}=5$~dB, $B=0.913+0.031j$ ($|B|=0.914$, $\arg B=1.9^\circ$) and $\sigma_\eta^2=0.0069$ (EVM $=\sigma_\eta/|B|=9.1\%$); at $3$~dB, $|B|=0.849$, $\sigma_\eta^2=0.0141$; at $8$~dB, $|B|=0.970$, $\sigma_\eta^2=0.0015$.

\subsection{Subcarrier Uniformity}
\label{sec:subcarrier_uniform}
\begin{proposition}[Approximately uniform DAFT-domain distortion]
\label{prop:uniform}
Under the Gaussian approximation of Lemma~\ref{lem:equimod}, the time-domain samples of an AFDM frame with i.i.d.\ symbols are pairwise uncorrelated with equal power, so the PA's Bussgang residual has covariance $\mathbf{C}_\eta\approx\sigma_\eta^2\mathbf{I}$, with an error that vanishes as $N$ grows under the Gaussian approximation. By unitarity, $\mathbf{A}\mathbf{C}_\eta\mathbf{A}^H\approx\sigma_\eta^2\mathbf{I}$: the gain $B$ and distortion power $\sigma_\eta^2$ are approximately equal on every chirp-subcarrier.
\end{proposition}
\begin{IEEEproof}[Proof sketch]
i.i.d.\ symbols give $\E\{\mathbf{x}\mathbf{x}^H\}=P_s\mathbf{I}$, which the unitary $\mathbf{A}^H$ preserves. Jointly Gaussian, uncorrelated samples are independent, and a memoryless nonlinearity maps independent inputs to independent residuals; the AFDM samples are only approximately Gaussian, so the statement is an approximation, not an exact identity. Monte-Carlo checks of the flat per-subcarrier distortion are in the appendices.
\end{IEEEproof}
Proposition~\ref{prop:uniform} uses only unitarity and the Gaussian approximation, so it also applies to OFDM; it does not apply to transforms whose samples are far from Gaussian, such as the identity (a single carrier). Equimodularity is what Theorem~\ref{thm:kernel} needs.

\subsection{Phase-Noise Kernel Equivalence}
\label{sec:pn_kernel}
\begin{theorem}[Kernel characterization of phase-noise equivalence]
\label{thm:kernel}
Let $\mathcal{D}$ be the diagonal unitary matrices and $\mathfrak{A}=\{\boldsymbol{\Lambda}_2\mathbf{F}\boldsymbol{\Lambda}_1:\boldsymbol{\Lambda}_1,\boldsymbol{\Lambda}_2\in\mathcal{D}\}$, which contains the DFT and the DAFT for every $(c_1,c_2)$. For a diagonal $\boldsymbol{\Phi}=\diag\{\Phi_n\}$ put $\mathbf{G}_{\mathbf{U}}=\mathbf{U}\boldsymbol{\Phi}\mathbf{U}^H$. Then:
\begin{enumerate}
\item[(a)] For any $\mathbf{A}\in\mathfrak{A}$, $\mathbf{G}_{\mathbf{A}}=\boldsymbol{\Lambda}_2\mathbf{G}_{\mathbf{F}}\boldsymbol{\Lambda}_2^H$, so the magnitude pattern and diagonal of $\mathbf{G}_{\mathbf{A}}$ equal the DFT's.
\item[(b)] For any unitary $\mathbf{U}$ whose entries all have modulus $1/\sqrt{N}$, $[\mathbf{G}_{\mathbf{U}}]_{kk}=J_0=\frac1N\sum_n\Phi_n$ for all $k$; if $|\Phi_n|=1$, the off-diagonal (leakage) power of every row is $1-|J_0|^2$.
\item[(c)] Conversely, if the diagonal of $\mathbf{G}_{\mathbf{U}}$ equals $J_0\mathbf{1}$ for every diagonal unitary $\boldsymbol{\Phi}$, then $\mathbf{U}$ has entries of equal modulus; and if the magnitude pattern equals the DFT's for every $\boldsymbol{\Phi}$, then $\mathbf{U}\in\mathfrak{A}$.
\end{enumerate}
\end{theorem}
\begin{IEEEproof}[Proof sketch]
(a) Diagonal conjugation preserves diagonals and magnitudes. (b) $[\mathbf{G}_{\mathbf{U}}]_{kk}=\sum_n|U_{kn}|^2\Phi_n$, and $\mathbf{G}_{\mathbf{U}}$ is unitary when $|\Phi_n|=1$. (c) See Appendix~\ref{sec:supp_theorem1}.
\end{IEEEproof}
Parts (a) and (b) restate classical OFDM results~\cite{pollet1995,armada2001,petrovic2005} for the affine--Fourier family. Part~(c) is a characterization: it shows that equimodularity is not merely sufficient, but it plays no role in the receiver or the experiments. Two consequences follow.

\begin{corollary}[Tracker portability]
\label{cor:tracker}
For any equimodular basis the CPE is the same $J_0$ on all subcarriers, so a decision-directed CPE tracker works identically in the DFT and DAFT domains.
\end{corollary}

\emph{Leakage.} No equimodular basis has less phase-noise leakage than OFDM, since total leakage power is $1-|J_0|^2$ in every row; AFDM's benefit lies in channel diversity, not in reduced leakage. For the Wiener model, with $\bar P=\E|f_{\mathrm{PA}}(s)|^2$ the mean PA output power and $r=e^{-\sigma_w^2/2}$,
\begin{equation}
\sigma_{\mathrm{ICI}}^2=\bar P\Bigl(1-\tfrac{1}{N^2}\bigl[N+2\textstyle\sum_{d=1}^{N-1}(N-d)r^d\bigr]\Bigr)\approx\bar P\,\sigma_w^2\,\tfrac{N^2-1}{6N}.
\label{eq:ici_afdm}
\end{equation}
At $\Delta f=100$~Hz and $N=256$ this is $\approx3\times10^{-3}$ at $\mathrm{IBO}=5$~dB; because $\sigma_w^2\propto1/N$ at fixed subcarrier spacing, it is nearly independent of $N$.

\subsection{Linewidth Threshold}
\label{sec:dfstar}
After CPE removal the leakage is data-dependent and acts like extra noise for a scalar-plus-tracker receiver. We define the heuristic threshold $\Delta f^*$ at which $\sigma_{\mathrm{ICI}}^2$ equals $\sigma_\eta^2$, the PA distortion that no scalar can remove; inverting~\eqref{eq:ici_afdm} gives $\Delta f^*=6N\sigma_\eta^2/[2\pi T_s(N^2-1)\bar P]$, which at $30$~kHz spacing is $550$, $234$ and $44$~Hz at $\mathrm{IBO}=3,5,8$~dB. It is an operating rule, not a fundamental limit, and it is the same for every equimodular waveform at equal $N$ and $T_s$. Our $100$~Hz default lies below it at $3$ and $5$~dB and above it at $8$~dB, where PA distortion is small anyway.

The Bussgang-matrix extension to per-antenna PAs in MIMO (block-diagonal gain, uniform distortion, and a cubic cross-antenna law) has been derived but is not included in this version; we do not evaluate MIMO here.

%% ---------------------------------------------------------------
\section{Experimental Evaluation}
\label{sec:experiments}
\subsection{Setup}
\label{sec:setup}
All experiments use one seeded Python/NumPy pipeline implementing Section~\ref{sec:system_model}: $N=256$ unless stated, the six-path fractional-Doppler channel, $\mathrm{IBO}=5$~dB, $100$~Hz phase noise, $30$~kHz spacing, and the sounding of Section~\ref{sec:baseline}, which is idealized (quasi-static, shared phase noise, uncharged overhead). SNR is the ``paper'' $E_b/N_0$ for four uncoded bits per symbol ($E_s/N_0=E_b/N_0+6$~dB), and it is \emph{not} reduced by the code rate in the coded runs, so every coded comparison is between receivers or waveforms at the same rate and SNR, never across rates. Uncoded percentages are means of per-channel reductions and need not equal one minus the ratio of the mean BERs shown in Table~\ref{tab:ber}. Every receiver sees the same bits, phase noise and thermal noise, and AFDM and OFDM see the same channel draws on 30 channels (seeds $13i+7$). For uncoded BER we report means of per-channel relative reductions, and for coded BLER one minus the ratio of mean BLERs, both with 95\% bootstrap confidence intervals over channels, and win counts (for BLER, wins are counted over \emph{informative} channels, those on which the two compared receivers do not tie, e.g.\ both error-free); $p$-values are two-sided Wilcoxon signed-rank tests, unadjusted, and $p=1.9\times10^{-9}$ is the smallest attainable at $n=30$ (a win on every channel). The preamble gain uses $4096$ known symbols ($K=16$ frames at $N=256$) and the M2M4 estimator the same number of unlabeled symbols. Code release is available upon acceptance.

For coding we use the industry-standard $(133,171)_8$, constraint-length-7 convolutional code at rates $1/2$ and $3/4$ (puncturing as in 802.11~\cite{yasuda1984}), one terminated codeword per symbol of $4N$ coded bits, a fixed random interleaver, max-log LLRs, and soft Viterbi decoding~\cite{viterbi1967}. Each receiver sets its LLR noise variance from the residual on the same 16 preamble frames. We run 64 codewords on each of 30 channels per point ($1920$ codewords), so block error rates (BLER) below about $1\%$ are rough.

\subsection{Gain Before Tracking}
\label{sec:order_results}
Table~\ref{tab:ber} lists the uncoded BER of every receiver at $24$~dB, and Table~\ref{tab:paired} the paired effects. At the default point ($\nu_{\max}=1$) the tracker alone cuts BER from $2.68\times10^{-2}$ to $1.36\times10^{-2}$ for AFDM (a $49\%$ reduction), a gain applied after the tracker reaches $1.15\times10^{-2}$, and the gain applied before it reaches $1.03\times10^{-2}$: relative to tracker-first, gain-first has $10.5\%$ lower BER, $95\%$ CI $[8.5,12.5]$, on $29/30$ channels. The effect is smaller for OFDM ($4.9\%$) and for lighter compression, as Section~\ref{sec:order} predicts.

\begin{table}[!t]
\renewcommand{\arraystretch}{1.05}
\setlength{\tabcolsep}{2.6pt}
\caption{Uncoded BER ($\times10^{-2}$), AFDM / OFDM, $N=256$, $\mathrm{IBO}=5$~dB, 24~dB, Mean of 30 Channels}
\label{tab:ber}
\centering
\footnotesize
\begin{tabular}{lccc}
\hline
\textbf{Receiver} & $\nu_{\max}=0.1$ & $\nu_{\max}=1$ & $\nu_{\max}=4$ \\
\hline
Plain LMMSE & 2.07 / 2.08 & 2.68 / 2.76 & 3.63 / 3.64 \\
Tracker only & 0.82 / 0.99 & 1.36 / 1.62 & 2.40 / 2.47 \\
Tracker $\to$ closed-form gain & 0.70 / 0.90 & 1.15 / 1.46 & 2.10 / 2.22 \\
\textbf{Closed-form gain $\to$ tracker} & \textbf{0.61} / 0.84 & \textbf{1.03} / 1.38 & \textbf{1.96} / 2.12 \\
Blind M2M4 gain $\to$ tracker & 0.62 / 0.85 & 1.03 / 1.38 & 1.94 / 2.10 \\
Preamble gain $\to$ tracker & \textbf{0.50} / 0.74 & \textbf{0.91} / 1.27 & \textbf{1.78} / 1.99 \\
\hline
Ideal hardware (LMMSE)$^{a}$ & 0.00 / 0.26 & 0.05 / 0.60 & 0.31 / 0.86 \\
\hline
\end{tabular}

\vspace{2pt}
\parbox{\columnwidth}{\scriptsize $^{a}$Linear PA, no phase noise, noiseless exact sounding, same 24~dB SNR.}
\end{table}

\begin{table*}[!t]
\renewcommand{\arraystretch}{1.05}
\setlength{\tabcolsep}{5pt}
\caption{Paired Uncoded Effects: Mean Relative BER Reduction (\%) [95\% CI] (Channels Won of 30), $24$~dB}
\label{tab:paired}
\centering
\footnotesize
\begin{tabular}{lcccc}
\hline
 & \multicolumn{2}{c}{\textbf{Gain-first vs.\ tracker-first}} & \textbf{Preamble vs.\ closed-form gain} & \textbf{AFDM vs.\ OFDM} \\
\textbf{Configuration} & AFDM & OFDM & AFDM & gain-first receiver \\
\hline
$\nu_{\max}=0.1$ & 12.1 [8.6, 15.9] (27) & 5.3 [2.7, 8.2] (25) & 12.4 [4.9, 21.2] (22) & 30.3 [21.1, 39.5] (26) \\
$\nu_{\max}=0.5$ & 11.1 [8.6, 13.7] (28) & 5.9 [3.6, 8.6] (26) & 11.2 [5.0, 18.4] (22) & 30.1 [22.6, 37.5] (28) \\
$\nu_{\max}=1$ (default) & 10.5 [8.5, 12.5] (29) & 4.9 [2.8, 7.3] (27) & 9.3 [4.2, 15.3] (23) & 30.8 [24.6, 37.0] (29) \\
$\nu_{\max}=2$ & 8.5 [6.8, 10.1] (29) & 4.6 [2.8, 6.6] (29) & 9.3 [3.8, 15.9] (25) & 22.8 [17.0, 28.9] (27) \\
$\nu_{\max}=4$ & 7.5 [5.9, 9.0] (29) & 5.0 [3.2, 7.3] (29) & 9.5 [4.0, 16.2] (28) & 13.8 [7.7, 20.6] (24) \\
\hline
IBO $3$~dB ($\nu_{\max}=1$) & 12.5 [10.6, 14.5] (30) & 9.4 [7.6, 11.4] (30) & 7.0 [4.0, 10.5] (28) & 6.1 [3.8, 8.5] (22) \\
IBO $8$~dB & 2.5 [$-$3.6, 6.6] (27) & 1.5 [$-$1.0, 4.1] (24) & 13.6 [6.2, 22.0] (23) & 50.8 [42.7, 58.7] (30) \\
$N=64$ & 5.3 [2.5, 7.9] (26) & 2.9 [0.2, 5.4] (24) & 21.6 [12.4, 32.2] (29) & 3.1 [$-$3.2, 9.5] (19) \\
$N=512$ & 12.4 [10.2, 14.5] (29) & 6.0 [4.4, 8.0] (29) & 6.0 [3.1, 9.4] (24) & 39.9 [33.2, 45.9] (29) \\
\hline
\end{tabular}
\end{table*}

Across Doppler the order effect for AFDM falls slowly from $12.1\%$ to $7.5\%$ and stays significant on $27$--$29$ of 30 channels; it is largest at moderate compression ($12.5\%$ at $3$~dB, $10.5\%$ at $5$~dB) and not significant at $8$~dB (CI includes zero), where compression is mild. It grows with the number of subcarriers ($5.3\%$ at $N=64$, $10.5\%$ at $256$, $12.4\%$ at $512$). In the toy setup of the appendices, the gain-first advantage of $19$--$39\%$ holds on 30 of 30 channels under every PA-on sounding model tried (independent or shared phase noise, clean or noisy, low or full pilot drive). A control separates compression from LMMSE shrinkage: with a linear PA the closed-form gain is unity, and a data-derived blind AGC applied before the tracker is $4.2\%$ \emph{worse} than after it, whereas with the PA on it is $16$--$35\%$ better (Appendix~\ref{sec:supp_nb50}). Together with the dependence on back-off, this points to PA compression as the cause. The effect is about twice as large for AFDM as for OFDM; Section~\ref{sec:afdm_vs_ofdm} shows that OCDM's order effect resembles AFDM's and DFT-s-OFDM's resembles OFDM's, which suggests a link to chirp-type spreading, but we have not isolated why.

\emph{Mechanism.} To test Section~\ref{sec:order}, we compare the two orders on identical frames (Table~\ref{tab:mech}). The error rate of the first hard decisions, before any tracking, is $8$--$10\%$ lower with gain-first (e.g.\ $7.4\%$ versus $8.1\%$ for AFDM at $\nu_{\max}=0.1$), and the tracker's phase estimate lies closer to the oracle common phase computed with the true symbols: $1.34^\circ$ versus $1.68^\circ$ RMS (AFDM, $\nu_{\max}=0.1$) and $1.64^\circ$ versus $2.04^\circ$ ($\nu_{\max}=1$), about $20\%$ lower, with the same pattern for OFDM. The first-decision improvement is modest and holds on only $19$--$20$ of 30 channels, so we regard the mechanism as supported, not proven.

\begin{table}[!t]
\renewcommand{\arraystretch}{1.05}
\setlength{\tabcolsep}{3pt}
\caption{Mechanism Check, Gain-First / Tracker-First, $N=256$, $\mathrm{IBO}=5$~dB, 24~dB, 30 Channels}
\label{tab:mech}
\centering
\footnotesize
\begin{tabular}{llccc}
\hline
 & & Phase error (deg) & First-decision & Final \\
 & & vs.\ oracle & SER (\%) & SER (\%) \\
\hline
AFDM & $\nu_{\max}=0.1$ & 1.34 / 1.68 & 7.35 / 8.12 & 2.38 / 2.74 \\
OFDM & $\nu_{\max}=0.1$ & 1.15 / 1.41 & 7.13 / 7.80 & 2.94 / 3.21 \\
AFDM & $\nu_{\max}=1$ & 1.64 / 2.04 & 9.48 / 10.53 & 4.09 / 4.55 \\
OFDM & $\nu_{\max}=1$ & 1.33 / 1.61 & 9.57 / 10.41 & 5.00 / 5.33 \\
\hline
\end{tabular}
\end{table}

\begin{figure*}[!t]
\centering
\includegraphics[width=\textwidth]{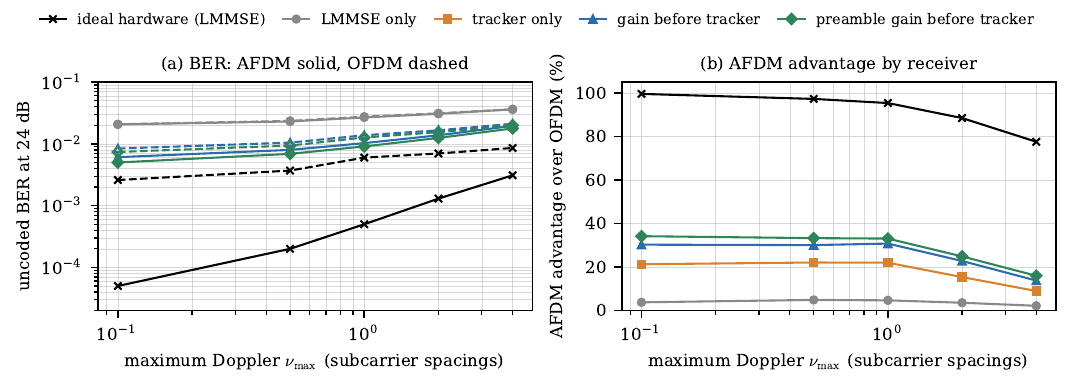}
\caption{Effect of maximum Doppler on uncoded BER at $24$~dB ($N=256$, $\mathrm{IBO}=5$~dB, 30 channels). (a) BER of AFDM (solid) and OFDM (dashed); the ideal-hardware AFDM point at $\nu_{\max}=0.1$ is below $5\times10^{-5}$ and is plotted there. (b) Relative BER advantage of AFDM over OFDM for each receiver; the advantage is $\ge 14\%$ for the gain-first receivers and does not grow with Doppler.}
\label{fig:doppler}
\end{figure*}

\subsection{Where the Gain Comes From}
\label{sec:gain_source}
The closed-form gain needs the PA model, but the gain \emph{value} is not unique. A blind M2M4 estimate computed from unlabeled symbols matches it within noise: its BER relative to the closed form is $+1.6\%$ ($[-1.8,5.4]$) for AFDM at the default point and its CI contains zero in every AFDM configuration we ran except $N=64$ (where it is $5.4\%$ better, $[0.5,10.7]$). The gain calibrated from the preamble does better in our idealized sounding, by $9.3\%$ ($[4.2,15.3]$) at the default point and $6$--$14\%$ across Table~\ref{tab:paired} at $N\ge256$; it needs 4096 known symbols on top of the sounding frames, so it is not a short preamble. The gain value is also not critical. Scaling the closed-form gain by $m\in[0.80,1.20]$ gives a flat optimum at $m=1.00$ for AFDM and OFDM at $\nu_{\max}=0.1$ and $1$; BER is within a few percent of it for $|m-1|\le0.05$ and $64$--$84\%$ worse at $m=0.80$, $32$--$62\%$ worse at $m=1.20$. The genie complex scalar that best maps the equalizer output to the true symbols has magnitude $0.962$ of the closed form (s.d.\ $0.003$--$0.004$ over channels), a systematic $3.8\%$ gap that we attribute to the pilots' own compression: a closed form that also accounts for it is $1.5\%$ from the genie value, but the BER basin is flat, so it never helped and we drop it. The genie scalar's phase has zero mean ($-0.1^\circ$) and a $5.6^\circ$ standard deviation across channels, consistent with a random static offset left by sounding with a shared phase-noise realization, which a complex preamble gain removes. Both this offset and the preamble's advantage depend on how sounding is done, and in-frame sounding could change them.

\emph{Comparison with pilots.} A PT-RS-style baseline inserts 16 known unit-power QPSK symbols per frame (every 16th position, $6.25\%$ overhead); every receiver in Table~\ref{tab:ptrs} is scored on the remaining 240 data symbols, and gain-first uses no pilots. A joint complex gain estimated by least squares from the pilots ties gain-first: gain-first has $9.2\%$ ($[-38.2,13.0]$) and $1.3\%$ ($[-17.4,11.1]$) \emph{higher} mean BER for AFDM at $\nu_{\max}=0.1$ and $1$, and $13.9\%$ and $0.6\%$ for OFDM, all intervals containing zero. Adding the tracker after the pilot gain, or using the closed-form gain with the pilot phase, lowers BER by $4$--$22\%$ relative to gain-first, with intervals that mostly contain zero. A pilot-phase-only estimate with no gain correction is $3.6$--$17\%$ worse than gain-first (significant at $\nu_{\max}=1$: $16.9\%$ and $9.1\%$), and the tracker alone is $14.5$--$29.8\%$ worse (28--30 of 30 channels). A genie per-frame scalar, an upper bound for any scalar method, is $7$--$31\%$ better. So the closed-form gain gives roughly what a pilot-based joint gain gives, without the overhead; getting the amplitude right is what matters, and the closed-form gain with the pilot phase is the best practical variant here.

\begin{table}[!t]
\renewcommand{\arraystretch}{1.05}
\setlength{\tabcolsep}{2.2pt}
\caption{Uncoded BER ($\times10^{-2}$) With 16 Pilots (6.25\% Overhead), 240 Data Symbols, $N=256$, $\mathrm{IBO}=5$~dB, 24~dB}
\label{tab:ptrs}
\centering
\scriptsize
\begin{tabular}{llcccccccc}
\hline
 & $\nu_{\max}$ & LMMSE & Tracker & \textbf{Gain-first} & Pilot & Pilot & Pilot & Closed & Genie \\
 & & & & \textbf{(no pilots)} & phase & gain & gain+trk & gain+ph & gain \\
\hline
AFDM & 0.1 & 2.08 & 0.83 & \textbf{0.60} & 0.70 & 0.59 & 0.54 & 0.53 & 0.51 \\
OFDM & 0.1 & 2.07 & 0.99 & \textbf{0.83} & 0.87 & 0.80 & 0.77 & 0.76 & 0.75 \\
AFDM & 1 & 2.70 & 1.38 & \textbf{1.05} & 1.25 & 1.08 & 1.00 & 0.99 & 0.95 \\
OFDM & 1 & 2.77 & 1.61 & \textbf{1.36} & 1.50 & 1.38 & 1.34 & 1.31 & 1.30 \\
\hline
\end{tabular}
\end{table}

The theory's contribution is therefore the structure (one scalar on all components, to be applied before the tracker), not the particular closed-form value.

\emph{HWI-aware LMMSE.} We also revisited the hardware-aware LMMSE of~\cite{sui2026hwi}, $\hat{\mathbf{s}}=\alpha\hat{\mathbf{H}}^H(\alpha^2\hat{\mathbf{H}}\hat{\mathbf{H}}^H+(\sigma_d^2+\sigma_n^2)\mathbf{I})^{-1}\mathbf{y}$ with a real Bussgang gain. Implemented with the regularizer $\sigma_d^2\mathbf{I}$ it loses by $21$--$78\%$ BER to the closed-form gain-first receiver in the toy setup, but that implementation mis-specifies the distortion covariance. With the covariance $(\alpha^2+\sigma_d^2)\hat{\mathbf{H}}\hat{\mathbf{H}}^H+\sigma_n^2\mathbf{I}$ and the bias removed it comes within about $2\%$ of the closed form, because the corrected receiver inverts the compression before any decision is made, i.e., it is also gain-first. We therefore do not claim a gain over a correctly specified HWI-aware LMMSE; what the tracker and the ordering add is the handling of phase noise (Appendix~\ref{sec:supp_nb50}).

\subsection{AFDM versus OFDM}
\label{sec:afdm_vs_ofdm}
With ideal hardware (last row of Table~\ref{tab:ber}), AFDM's BER is $77$--$99.6\%$ below OFDM's (mean of per-channel reductions) at every Doppler; this is the diversity advantage of spreading. Hardware impairments erode it: with a plain LMMSE receiver the edge is only $4.7\%$ at $\nu_{\max}=1$, and it recovers to $22.0\%$ with a tracker, $30.8\%$ gain-first, and $33.1\%$ with the preamble gain ($[27.0,39.0]$; $p=3.7\times10^{-9}$ for gain-first). So the spreading advantage is mostly unrealized without a hardware-aware receiver.

The edge does not grow with Doppler: it is $30\%$ already at $\nu_{\max}=0.1$ and falls to $14\%$ at $\nu_{\max}=4$, so these data do not support choosing AFDM because Doppler is large. It also depends on the operating point: at $N=64$, the size of our earlier toy setup, it is not significant ($3.1\%$, CI includes zero), at $N=512$ it is $39.9\%$; at IBO $3$~dB, where the PA's EVM violates the NR limit, it shrinks to $6.1\%$; at $8$~dB it is $50.8\%$.

\emph{Spreading controls.} We attribute the edge to spreading: with six delayed paths, any transform that spreads each symbol across the frequency-selective channel gives delay diversity even at negligible Doppler, whereas OFDM's per-subcarrier fades are not averaged by an uncoded one-symbol decision. To test this we added OCDM (the DAFT with $c_1=c_2=1/(2N)$) and DFT-s-OFDM with full-band spreading, a single carrier with a cyclic prefix, on identical channels and hardware (Table~\ref{tab:spread}). With ideal hardware all three spreading waveforms are far better than OFDM ($95$--$99.6\%$, $60$--$95\%$ and $58$--$93\%$ lower BER for AFDM, OCDM and DFT-s-OFDM). With real hardware, OCDM obtains $26.2\%$ (gain-first, $\nu_{\max}=0.1$) and $12.3\%$ ($\nu_{\max}=1$) against AFDM's $30.3\%$ and $30.8\%$, so AFDM's chirp rate may help once Doppler is large, but DFT-s-OFDM obtains $30.8\%$ and $26.6\%$ and its mean BER with the preamble gain is $19\%$ below AFDM's at $\nu_{\max}=0.1$ and $2.5\%$ below at $\nu_{\max}=1$. So under our PA model DFT-s-OFDM matches AFDM at equal IBO. Two cautions apply: its lower peak-to-average power ratio means it is compressed less at equal back-off, and the closed-form gain assumes a Gaussian input and is the wrong constant for it, so we compare it through the preamble columns; and we did not compute paired intervals for AFDM versus DFT-s-OFDM. The order effect is larger for the chirp waveforms ($12.1$ and $10.5\%$ for AFDM, $11.0$ and $7.1\%$ for OCDM) than for OFDM ($5.3$ and $4.9\%$) or DFT-s-OFDM ($4.7\%$, CI includes zero, and $4.6\%$).

\begin{table*}[!t]
\renewcommand{\arraystretch}{1.05}
\setlength{\tabcolsep}{5pt}
\caption{Spreading Controls, Real Hardware: Uncoded BER ($\times10^{-2}$) at 24~dB, Edge Over OFDM for the Gain-First Receiver [95\% CI] (Channels Won of 30), and Gain-First vs.\ Tracker-First Order Effect (\%)}
\label{tab:spread}
\centering
\footnotesize
\begin{tabular}{llcccccc}
\hline
$\nu_{\max}$ & Waveform & BER gain-first & BER preamble & Edge over OFDM (\%) & Order effect (\%) & Ideal-hardware edge (\%) \\
\hline
0.1 & AFDM & 0.61 & 0.50 & 30.3 [21.1, 39.5] (26) & 12.1 [8.6, 15.9] & 99.6 \\
 & OCDM & 0.66 & 0.55 & 26.2 [17.3, 35.0] (26) & 11.0 [6.7, 15.1] & 94.6 \\
 & DFT-s-OFDM & 0.57 & 0.40 & 30.8 [14.7, 44.7] (25) & 4.7 [$-$2.0, 10.3] & 93.3 \\
 & OFDM & 0.84 & 0.74 & --- & 5.3 [2.7, 8.2] & --- \\
\hline
1 & AFDM & 1.03 & 0.91 & 30.8 [24.6, 37.0] (29) & 10.5 [8.5, 12.5] & 95.3 \\
 & OCDM & 1.23 & 1.12 & 12.3 [6.7, 17.5] (25) & 7.1 [5.2, 9.1] & 59.9 \\
 & DFT-s-OFDM & 1.05 & 0.89 & 26.6 [18.8, 34.2] (27) & 4.6 [2.4, 6.7] & 58.0 \\
 & OFDM & 1.38 & 1.27 & --- & 4.9 [2.8, 7.3] & --- \\
\hline
\end{tabular}
\end{table*}

\subsection{Coded Performance}
\label{sec:coded}
Channel coding lets OFDM recover diversity across subcarriers, so the uncoded edge need not survive. Fig.~\ref{fig:bler}, Table~\ref{tab:codedpaired} and Table~\ref{tab:snr10} give the result. At rate $3/4$, AFDM keeps and grows its edge over OFDM (the spreading controls of Section~\ref{sec:afdm_vs_ofdm} were run uncoded only, so we do not know whether DFT-s-OFDM would match it here): at $\mathrm{IBO}=5$~dB the SNR needed for $10\%$ BLER with gain-first is $18.3$~dB for AFDM and $21.0$~dB for OFDM; at $8$~dB it is $17.1$ versus $19.3$~dB; at $3$~dB OFDM does not reach $10\%$ by $24$~dB, AFDM does ($21.3$~dB). AFDM's BLER is lower by $33.5\%$ at $16$~dB, $59.3\%$ at $20$~dB and $52.8\%$ at $24$~dB for $\mathrm{IBO}=5$~dB (gain-first), with all CIs far above zero. At rate $1/2$ the SNR gap at $10\%$ BLER is only $0$--$0.5$~dB in every configuration (AFDM's BLER is still $11$--$18\%$ lower at $12$~dB, but the waterfall is steep, and the CI at $\mathrm{IBO}=8$~dB includes zero), so a strong low-rate code recovers most of the diversity that OFDM lacked; AFDM's uncoded advantage should not be assumed to carry over to a well-coded link.

\begin{figure*}[!t]
\centering
\includegraphics[width=\textwidth]{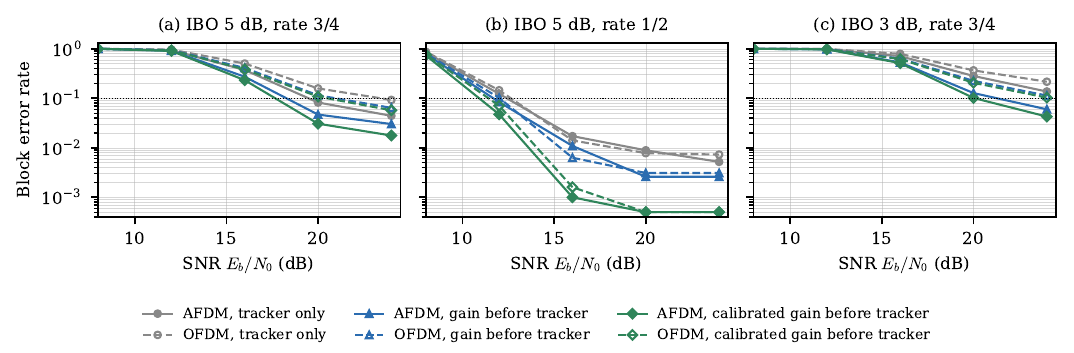}
\caption{BLER versus SNR ($E_b/N_0$ of the uncoded 16-QAM) for the $(133,171)_8$ convolutional code, 30 channels $\times$ 64 codewords, $N=256$, $\nu_{\max}=1$. Solid: AFDM; dashed, open markers: OFDM. Zero errors are plotted at the $5\times10^{-4}$ floor. Dotted line: $10\%$ BLER.}
\label{fig:bler}
\end{figure*}

\begin{table*}[!t]
\renewcommand{\arraystretch}{1.05}
\setlength{\tabcolsep}{5pt}
\caption{Paired Coded Effects: Mean Relative BLER Reduction (\%) [95\% CI] (Wins / Informative Channels)}
\label{tab:codedpaired}
\centering
\footnotesize
\begin{tabular}{llcccc}
\hline
 & & \multicolumn{2}{c}{\textbf{Gain-first vs.\ tracker-first}} & \textbf{Preamble vs.\ closed form} & \textbf{AFDM vs.\ OFDM} \\
\textbf{IBO, rate} & \textbf{SNR} & AFDM & OFDM & AFDM & gain-first receiver \\
\hline
5 dB, 3/4 & 16 dB & 9.2 [6.5, 12.1] 23/25 & 3.3 [1.2, 5.4] 17/25 & 14.7 [7.4, 22.7] 18/24 & 33.5 [25.4, 41.9] 27/27 \\
 & 20 dB & 24.4 [15.4, 32.6] 15/16 & 5.6 [$-$1.2, 12.6] 11/18 & 34.4 [18.7, 46.3] 11/12 & 59.3 [46.8, 72.7] 26/28 \\
 & 24 dB & 15.9 [2.2, 25.0] 9/12 & 11.5 [1.9, 20.6] 8/11 & 41.4 [21.1, 59.6] 11/12 & 52.8 [33.9, 68.7] 18/21 \\
3 dB, 3/4 & 16 dB & 9.7 [6.0, 14.1] 24/26 & 3.7 [2.0, 5.6] 21/25 & 3.0 [0.8, 5.5] 13/18 & 16.4 [11.1, 22.4] 26/28 \\
 & 20 dB & 23.9 [18.2, 30.3] 25/26 & 14.6 [10.6, 18.6] 22/24 & 20.4 [8.7, 30.4] 16/20 & 43.3 [34.0, 53.0] 29/30 \\
 & 24 dB & 27.7 [19.0, 35.9] 19/19 & 13.8 [6.6, 20.4] 15/19 & 28.7 [13.8, 43.4] 13/16 & 47.2 [36.6, 57.9] 23/25 \\
8 dB, 3/4 & 16 dB & 4.5 [0.3, 9.0] 11/16 & 0.0 [$-$2.0, 1.9] 9/17 & 22.6 [12.7, 31.9] 16/17 & 47.7 [37.6, 58.2] 30/30 \\
 & 20 dB & 11.8 [$-$2.6, 18.8] 5/7 & 5.0 [0.0, 10.2] 9/11 & 46.7 [17.8, 66.2] 9/10 & 60.8 [45.3, 76.8] 21/22 \\
 & 24 dB & 13.3 [$-$11.8, 29.4] 4/7 & 7.9 [$-$4.2, 17.9] 8/13 & 56.4 [30.8, 75.0] 9/10 & 52.4 [30.6, 72.1] 19/23 \\
\hline
5 dB, 1/2 & 12 dB & 13.0 [5.6, 21.9] 15/17 & 9.3 [3.8, 15.1] 13/16 & 43.5 [27.0, 58.0] 19/20 & 17.9 [7.8, 29.9] 18/25 \\
3 dB, 1/2 & 12 dB & 15.8 [12.0, 20.5] 23/24 & 13.5 [9.8, 17.4] 21/22 & 30.0 [17.7, 40.8] 18/21 & 11.1 [1.5, 21.2] 16/25 \\
8 dB, 1/2 & 12 dB & 7.3 [1.1, 12.7] 5/6 & 5.4 [$-$3.8, 12.6] 7/11 & 59.4 [38.8, 73.1] 16/18 & 17.2 [$-$2.6, 37.8] 10/17 \\
\hline
\end{tabular}
\end{table*}

\begin{table}[!t]
\renewcommand{\arraystretch}{1.05}
\setlength{\tabcolsep}{3pt}
\caption{SNR (dB) for $10\%$ BLER, AFDM / OFDM}
\label{tab:snr10}
\centering
\footnotesize
\begin{tabular}{llcccc}
\hline
\textbf{IBO} & \textbf{Rate} & Tracker & Gain after & Gain before & Preamble \\
\hline
3 dB & 1/2 & 14.2 / 14.3 & 13.4 / 13.3 & 12.9 / 12.9 & 12.2 / 12.5 \\
3 dB & 3/4 & $>$24 / $>$24 & 22.9 / $>$24 & 21.3 / $>$24 & 20.1 / $>$24 \\
5 dB & 1/2 & 12.4 / 12.6 & 11.9 / 12.2 & 11.7 / 12.0 & 10.9 / 11.4 \\
5 dB & 3/4 & 19.5 / 23.4 & 18.8 / 21.5 & 18.3 / 21.0 & 17.7 / 20.5 \\
8 dB & 1/2 & 11.2 / 11.5 & 11.1 / 11.3 & 11.0 / 11.2 & 10.1 / 10.5 \\
8 dB & 3/4 & 17.5 / 19.6 & 17.3 / 19.4 & 17.1 / 19.3 & 16.4 / 18.8 \\
\hline
\end{tabular}
\end{table}

The order effect survives coding. At rate $3/4$ and $\mathrm{IBO}=5$~dB, gain-first lowers AFDM's BLER by $9.2\%$ at $16$~dB (23 of 25 informative channels) and $24.4\%$ at $20$~dB (15 of 16), which is $0.5$~dB at $10\%$ BLER (at rate $1/2$ it is $0.2$~dB at $5$~dB back-off and $0.5$~dB at $3$~dB); at $3$~dB and rate $3/4$ it is $1.6$~dB ($9.7$, $23.9$ and $27.7\%$ at $16$, $20$, $24$~dB). At $8$~dB it is $0.2$~dB and not significant at higher SNR. For OFDM the effect is smaller in every row of Table~\ref{tab:codedpaired}. Against the tracker alone, gain-first saves $1.2$~dB at $\mathrm{IBO}=5$~dB.

The preamble gain is worth a further $0.6$~dB at $\mathrm{IBO}=5$~dB ($1.2$~dB at $3$~dB) and lowers AFDM's BLER by $14.7$, $34.4$ and $41.4\%$ at $16$, $20$ and $24$~dB relative to the closed form. In our idealized sounding it also removes an error floor: at rate $1/2$ and $24$~dB the tracker-only, tracker-first and gain-first receivers floor at $0.3$--$1\%$ BLER, while the preamble receiver has no errors in $1920$ AFDM codewords at $\mathrm{IBO}=5$~dB. The floor may be an artifact of the quasi-static sounding rather than a property of the receivers. The preamble gain's benefit does not shrink with a lighter PA (it is larger at $8$~dB than at $5$~dB), which is consistent with a PA-independent cause such as the static sounding phase offset of Section~\ref{sec:gain_source}; we did not isolate it, and it needs confirmation with in-frame sounding.

%% ---------------------------------------------------------------
\section{Conclusion and Limitations}
\label{sec:conclusion}
We showed that AFDM's equimodular transform reduces PA nonlinearity to one gain on all components and phase noise to a common rotation, so OFDM's phase tracker applies unchanged, and that the gain should be corrected first. In a six-path, fractional-Doppler, 256-subcarrier setting with an NR-compliant PA, gain-first lowers uncoded BER by $7.5$--$12\%$ for AFDM and about $5\%$ for OFDM (not significant at $8$~dB back-off) and, at rate $3/4$, gives $0.5$~dB at $10\%$ BLER; its first decisions are $8$--$10\%$ better and its phase estimate about $20\%$ closer to the oracle. A blind moment estimate matches the closed form, and the closed form, with no pilots, performs about as well as a pilot-based joint complex gain that spends $6.25\%$ of the frame; the best practical variant combines the closed-form gain with a pilot phase. AFDM keeps a $14$--$31\%$ uncoded BER edge over OFDM and $2$--$3$~dB at rate $3/4$ under identical hardware, but the edge does not grow with Doppler, nearly vanishes at rate $1/2$, and is mostly spreading: OCDM and DFT-s-OFDM obtain much of it, and DFT-s-OFDM matches AFDM at equal back-off. We therefore do not claim an AFDM advantage over other spreading waveforms; the portable result is the gain-before-tracking correction, which applies to AFDM, OCDM and OFDM alike.

\emph{Limitations.} Results are simulated, with one code family (convolutional; LDPC and polar codes are untested), 64 codewords per channel, and $\nu_{\max}=1$, $N=256$ for the coded study. The PA is memoryless and symbol-rate, without oversampling, spectral regrowth or memory. Sounding is idealized: $N$ pilot frames that share one phase-noise realization, quasi-static, with overhead not charged, so the preamble-gain advantage, the static-offset explanation and the removal of the rate-1/2 error floor may partly be artifacts of it, and a practical in-frame estimator with continuous Doppler phase is left open. The order effect is measured at one PT-RS density, and the DFT-s-OFDM and OCDM controls are uncoded, unpaired against AFDM and, for DFT-s-OFDM, confounded by PAPR. The reason the order effect is larger for chirp waveforms than for OFDM is untested. We report BER and BLER at fixed back-off, not total degradation, and the PA has no oversampling or spectral regrowth. We do not compare with digital predistortion~\cite{dpd_survey2010} or OTFS~\cite{hadani2017otfs}. Results from an earlier toy setup on a learned refinement, MIMO and 3GPP TDL-C channels used the tracker-first order and are not included in this version.

%% ---------------------------------------------------------------
\section*{Acknowledgment}
The authors acknowledge the use of Anthropic's Claude to assist with language editing and with preparing simulation notebooks. All core ideas, methodology and results were developed by the authors, who reviewed and verified the final text in full for technical accuracy.

\makeatletter
\renewcommand{\IEEEbibitemsep}{0pt plus 0.3pt}
\makeatother

%% ================================================================
%% Supplementary material
%% ================================================================
\clearpage
\newgeometry{left=1.05in,right=1.05in,top=0.85in,bottom=1in,headsep=0.25in}
\onecolumn
\setlength{\columnwidth}{\textwidth}\setlength{\hsize}{\textwidth}\setlength{\linewidth}{\textwidth}
\appendices
\setcounter{figure}{0}\setcounter{table}{0}
\renewcommand{\thefigure}{S\arabic{figure}}
\renewcommand{\thetable}{S\arabic{table}}
\begin{center}
{\LARGE Supplementary Material}\\[8pt]
{\large Gain Before Tracking: Compensating Power-Amplifier Nonlinearity\\ and Phase Noise in AFDM Uplinks}
\end{center}
\vspace{10pt}
The appendices are organized as follows.
\begin{itemize}
\item Appendices~\ref{sec:supp_lemma1}--\ref{sec:supp_prop1}: proofs of Lemma~\ref{lem:equimod}, Theorem~\ref{thm:kernel}(c) and Proposition~\ref{prop:uniform}.
\item Appendix~\ref{sec:supp_dfstar}: derivation of the linewidth threshold $\Delta f^*$.
\item Appendix~\ref{sec:supp_realistic}: protocol and configuration of the experiments, and the sounding-model and linear-PA controls behind Section~\ref{sec:order_results}.
\item Appendix~\ref{sec:supp_diag}: additional diagnostics, namely the gain sweep, the mechanism check, the PT-RS-style baseline and the spreading controls (OCDM, DFT-s-OFDM).
\end{itemize}
An extended supplement (an earlier toy-setup study with a learned refinement, OAMP-Net2 and DPD baselines, MIMO and LDPC results, 3GPP TDL-C validation, and the matrix Bussgang theorem for per-antenna PAs) is available from the authors.
All numerical claims were produced by the simulation code, which is available upon acceptance.  

%% ==================================================================
\FloatBarrier
\section{Proof of Lemma~1: Equimodularity of the DAFT Basis}
\label{sec:supp_lemma1}
%% ==================================================================

\textbf{Lemma~1 restated.}
Every entry of $\mathbf{A}^H$ has modulus $1/\sqrt{N}$. Consequently
each time-domain sample of an AFDM frame is a sum of $N$ equal-weight
independent terms, and, as approximations for large $N$, \emph{(i)} the
distribution of the AFDM time-domain signal is close to complex Gaussian,
with the Berry--Esseen error decaying as $O(N^{-1/2})$ (the constant
depends on the phases of the weights and is not claimed to equal OFDM's);
\emph{(ii)} the marginal distribution of every sample, and hence the PAPR
statistics, are approximately OFDM's; \emph{(iii)} the Bussgang decomposition
applies to AFDM with approximately the accuracy it has for OFDM. The
empirical check below, not the argument, is what supports (ii)--(iii).

\begin{IEEEproof}
By the DAFT factorisation, $\mathbf{A}^H = \boldsymbol{\Lambda}_{c_1}^H
\mathbf{F}^H\boldsymbol{\Lambda}_{c_2}^H$ with
$\boldsymbol{\Lambda}_{c_1}^H = \mathrm{diag}(e^{+j2\pi c_1 n^2})$ and
$\boldsymbol{\Lambda}_{c_2}^H = \mathrm{diag}(e^{+j2\pi c_2 k^2})$ both
diagonal unitary and $\mathbf{F}^H$ the unitary IDFT. Hence
\begin{equation*}
\begin{aligned}
|[\mathbf{A}^H]_{nk}|
&= |e^{j2\pi c_1 n^2}|\cdot|[\mathbf{F}^H]_{nk}|\cdot|e^{j2\pi c_2 k^2}|\\
&= 1\cdot\frac{1}{\sqrt{N}}\cdot 1 = \frac{1}{\sqrt{N}},
\end{aligned}
\end{equation*}
since $|[\mathbf{F}^H]_{nk}| = 1/\sqrt{N}$ for every entry of the
unitary IDFT matrix and both diagonal factors are unimodular.

Each time-domain sample $s_n = \sum_k [\mathbf{A}^H]_{nk} x_k$ is
therefore a sum of $N$ terms of equal modulus $1/\sqrt{N}$ weighting
i.i.d.\ (or independent, for fixed QAM symbols) inputs $x_k$ --
structurally identical to the OFDM/IDFT superposition. Items (i)--(ii)
follow immediately by the classical Berry--Esseen / central-limit-theorem
argument used for the IDFT modulator~\cite{ochiai2001papr,dardari2000theoretical}:
the moduli of the weights are the same for $\mathbf{A}^H$ and
$\mathbf{F}^H$, and the central-limit argument gives a Gaussian limit for
both, although the Berry--Esseen constant for the real and imaginary parts
depends on the weight phases, which differ. Item (iii) is then an
approximation: the Bussgang decomposition is accurate to the extent that
the per-sample distribution is close to circular Gaussian, which is
close for both waveforms at the $N$ used here.
\end{IEEEproof}

\textbf{Empirical validation.} Simulation confirms the Gaussian
approximation numerically: the fourth-moment ratio is $1.999$ for
AFDM (OFDM: $2.002$; circular Gaussian: $2.000$), and mean PAPR is
$6.73$~dB (OFDM: $6.74$~dB). Over $10^4$ frames, AFDM's $1\%/0.1\%$
PAPR percentiles are $9.35/10.31$~dB (OFDM: $9.41/10.34$~dB), and the
maxima ($10.89$ vs.\ $11.66$~dB) show no heavier tail for AFDM---the
full CCDF matches OFDM closely, with no loss of Bussgang accuracy.

%% ==================================================================
\FloatBarrier
\section{Proof of Theorem~1(c): Converses of the Kernel
Characterisation}
\label{sec:supp_theorem1}
%% ==================================================================

\textbf{Theorem~1(c) restated.}
Let $\mathbf{U}$ be an $N\times N$ unitary matrix and write
$\mathbf{G}_{\mathbf{U}} \triangleq \mathbf{U}\boldsymbol{\Phi}
\mathbf{U}^H$ for diagonal $\boldsymbol{\Phi}$.
\begin{enumerate}
\item[(c1)] If $\mathrm{diag}(\mathbf{G}_{\mathbf{U}}) = J_0\mathbf{1}$
  for every diagonal unitary $\boldsymbol{\Phi}$, then $\mathbf{U}$
  is equimodular ($|U_{k,n}| = 1/\sqrt{N}$ for all $k,n$).
\item[(c2)] If $|[\mathbf{G}_{\mathbf{U}}]_{kl}| =
  |[\mathbf{G}_{\mathbf{F}}]_{kl}|$ for all $k,l$ and all diagonal
  unitary $\boldsymbol{\Phi}$, then
  $\mathbf{U} \in \mathfrak{A}$, i.e.\
  $\mathbf{U} = \boldsymbol{\Lambda}_2 \mathbf{F} \boldsymbol{\Lambda}_1$
  for diagonal unitaries $\boldsymbol{\Lambda}_1, \boldsymbol{\Lambda}_2$.
\end{enumerate}

\begin{IEEEproof}
\textbf{Proof of (c1) — single-site test.}
Choose the single-site impairment
$\boldsymbol{\Phi}^{(n_0,\alpha)} =
\mathbf{I} + (e^{j\alpha}-1)\mathbf{e}_{n_0}\mathbf{e}_{n_0}^H$,
i.e.\ $\Phi_n = 1+(e^{j\alpha}-1)\delta_{n,n_0}$.
Then for any $k$,
\begin{equation}
[\mathbf{G}_{\mathbf{U}}]_{kk}
  = \sum_{n}|U_{kn}|^2\Phi_n
  = 1 + (e^{j\alpha}-1)|U_{k,n_0}|^2.
\label{eq:single_site_diag}
\end{equation}
The hypothesis $\mathrm{diag}(\mathbf{G}_{\mathbf{U}})=J_0\mathbf{1}$
requires this to equal $\frac{1}{N}\sum_n\Phi_n =
1+(e^{j\alpha}-1)/N$ for \emph{every} $k$ and every $\alpha$.
Since $(e^{j\alpha}-1)\neq 0$ for $\alpha\notin 2\pi\mathbb{Z}$,
we need $|U_{k,n_0}|^2 = 1/N$ for every $k$ and $n_0$.
As $n_0$ ranges over $\{0,\ldots,N-1\}$ this gives equimodularity.
$\hfill\square$

\textbf{Proof of (c2) — two-site phase condition and cocycle
factorisation.}
Fix $n_1\neq n_2$ and take
\begin{equation*}
\Phi_n = 1 + (e^{j\alpha}-1)\delta_{n,n_1}
           + (e^{j\beta} -1)\delta_{n,n_2},
\end{equation*}
so $a=e^{j\alpha}-1$, $b=e^{j\beta}-1$. For $k\neq l$,
\begin{equation}
[\mathbf{G}_{\mathbf{U}}]_{kl}
  = \sum_n U_{kn}\Phi_n U_{ln}^*
  = [\mathbf{G}_{\mathbf{U}}^{(\Phi\equiv 1)}]_{kl}
    + a\,z_1 + b\,z_2
\label{eq:two_site_off}
\end{equation}
where $z_i = U_{k,n_i}U_{l,n_i}^*$.
For the DFT baseline the same expression holds with
$w_i = \frac{1}{N}e^{-j2\pi(k-l)n_i/N}$.
The hypothesis requires $|az_1+bz_2|=|aw_1+bw_2|$ for every
$\alpha,\beta\in\mathbb{R}$ \emph{varied independently} (not merely
along the diagonal $\beta=\alpha$, which forces $b=a$ and collapses
the two free parameters to one). Squaring gives
\begin{equation*}
\begin{aligned}
&|a|^2|z_1|^2+|b|^2|z_2|^2+2\,\mathrm{Re}(a\bar{b}\,z_1\bar{z}_2)\\
&\qquad = |a|^2|w_1|^2+|b|^2|w_2|^2+2\,\mathrm{Re}(a\bar{b}\,w_1\bar{w}_2),
\end{aligned}
\end{equation*}
required to hold for every such $(\alpha,\beta)$. We extract the
content in three steps that vary $\alpha$ and $\beta$ independently.

\emph{Step 1 ($|z_1|=|w_1|$).} Fix $\beta=0$ (so $b=0$) and let
$\alpha$ range over $\mathbb{R}\setminus2\pi\mathbb{Z}$ (so
$a\neq0$): the identity reduces to $|a|^2|z_1|^2=|a|^2|w_1|^2$,
giving $|z_1|=|w_1|=1/N$.

\emph{Step 2 ($|z_2|=|w_2|$).} Symmetrically, fix $\alpha=0$ and
vary $\beta$ to get $|z_2|=|w_2|=1/N$.

\emph{Step 3 (the cross term).} With Steps 1--2 the $|a|^2$ and
$|b|^2$ terms cancel identically on both sides, leaving
\begin{equation*}
\mathrm{Re}\big(a\bar{b}\,\Delta\big)=0,
\qquad
\Delta \triangleq z_1\bar{z}_2-w_1\bar{w}_2,
\end{equation*}
for every $\alpha,\beta$ varied independently. Evaluating at the
four independently-chosen pairs
$(\alpha,\beta)\in\{(\pi/2,\pi/2),(\pi/2,\pi),(\pi,\pi/2),(\pi,\pi)\}$,
for which $a,b\in\{j-1,\,-2\}$, gives
$a\bar{b}\in\{2,\ 2-2j,\ 2+2j,\ 4\}$ respectively. The pairs
$a\bar{b}=2$ and $a\bar{b}=2+2j$ yield $\mathrm{Re}(\Delta)=0$ and
$\mathrm{Re}(\Delta)-\mathrm{Im}(\Delta)=0$; together these force
$\Delta=0$, i.e.\ $z_1\bar{z}_2=w_1\bar{w}_2$:
\begin{equation}
U_{k,n_1}U_{l,n_1}^*\,\overline{U_{k,n_2}U_{l,n_2}^*}
  = \frac{1}{N^2}e^{-j2\pi(k-l)(n_1-n_2)/N}.
\label{eq:cocycle_cond}
\end{equation}
This holds for all $k\neq l$ and all pairs $(n_1,n_2)$; Steps 1--2
also reprove (c1)'s magnitude conclusion along the way.

\textbf{Cocycle factorisation.}
Define $\mu_k(n) = \sqrt{N}\,U_{k,n}\,e^{j2\pi kn/N}$.
From (c1), $|\mu_k(n)|=1$ for all $k,n$.
Equation~\eqref{eq:cocycle_cond} with fixed $n_2=0$ gives
\begin{equation*}
\mu_k(n_1)\overline{\mu_l(n_1)} = \mu_k(0)\overline{\mu_l(0)}
\quad\forall\, k\neq l,\; n_1.
\end{equation*}
Setting $l=0$ as reference: $\mu_k(n)/\mu_0(n)$ is independent
of $n$ for every $k$.  Write $c_k = \mu_k(n)/\mu_0(n)$; then
\begin{equation*}
U_{k,n} = \frac{1}{\sqrt{N}}\,e^{-j2\pi kn/N}\,c_k\,\mu_0(n)
  = [\boldsymbol{\Lambda}_2\mathbf{F}\boldsymbol{\Lambda}_1]_{k,n}
\end{equation*}
with $\boldsymbol{\Lambda}_2=\mathrm{diag}\{c_k\}$ and
$\boldsymbol{\Lambda}_1=\mathrm{diag}\{\mu_0(n)\}$, both diagonal
unitary.  Hence $\mathbf{U}\in\mathfrak{A}$.
$\hfill\blacksquare$
\end{IEEEproof}

\begin{remark}
The converse (c2) identifies the affine-Fourier family $\mathfrak{A}$
as the \emph{unique} class of unitary modulators for which the
full entrywise pattern of phase-noise leakage magnitudes coincides
with OFDM's.  Any unitary modulator outside $\mathfrak{A}$ has
at least one off-diagonal leakage entry with a different magnitude
than the DFT produces, so its CPE tracker and ICI power formula
must be re-derived.  The DAFT and OCDM both belong to $\mathfrak{A}$
(diagonal unitaries composed with the DFT), while OTFS in the
symplectic domain does not in general.
\end{remark}

%% ==================================================================
\FloatBarrier
\section{Proof of Proposition~1: Subcarrier-Uniformity of
DAFT-Domain Distortion}
\label{sec:supp_prop1}
%% ==================================================================

\textbf{Proposition~1 restated.}
For i.i.d.\ $M$-QAM symbols, $\mathbb{E}\{\mathbf{s}\mathbf{s}^H\}
= P_s\mathbf{I}$ by unitarity of $\mathbf{A}^H$. Under the Gaussian approximation of Lemma~1, a memoryless
nonlinearity acting on pairwise-uncorrelated jointly-Gaussian samples
yields pairwise-uncorrelated Bussgang residuals, so
$\mathbf{C}_\eta \approx \sigma_\eta^2\mathbf{I}$ in the time domain and,
again by unitarity, $\mathbf{A}\mathbf{C}_\eta\mathbf{A}^H
\approx \sigma_\eta^2\mathbf{I}$ in the DAFT domain. The Bussgang gain and
distortion power are therefore approximately identical on every
chirp-subcarrier; the statement is an approximation, validated numerically
below, not an exact identity.

\begin{IEEEproof}
I.i.d.\ QAM symbols $\mathbf{x}$ satisfy
$\mathbb{E}\{\mathbf{x}\mathbf{x}^H\} = P_s\mathbf{I}$ by
definition. Since $\mathbf{A}^H$ is unitary, the time-domain frame
$\mathbf{s} = \mathbf{A}^H\mathbf{x}$ satisfies
$\mathbb{E}\{\mathbf{s}\mathbf{s}^H\} = \mathbf{A}^H
\mathbb{E}\{\mathbf{x}\mathbf{x}^H\}\mathbf{A}
= P_s\mathbf{A}^H\mathbf{A} = P_s\mathbf{I}$, i.e.\ the time-domain
samples are pairwise uncorrelated with equal power $P_s$, exactly as
for OFDM (this is also Lemma~1's Gaussian-approximation claim
applied to the second moment alone).

For pairwise-uncorrelated jointly-Gaussian inputs (the Gaussian approximation of Lemma~1; the AFDM samples are only approximately Gaussian, so the conclusions below hold approximately), a memoryless
nonlinearity $f$ acting sample-wise produces pairwise-uncorrelated
outputs: by the Bussgang theorem, $f(s_n) = B s_n + \eta_n$ with
$\mathbb{E}\{\eta_n s_m^*\} = 0$ for all $m$, and for $n\neq m$ the
cross term $\mathbb{E}\{\eta_n\eta_m^*\}$ vanishes because $s_n$ and
$s_m$ are independent (jointly Gaussian and uncorrelated implies
independent) and $\eta_n$ depends only on $s_n$
\cite{demir2021distortion}. Hence
$\mathbf{C}_\eta = \mathbb{E}\{\boldsymbol{\eta}\boldsymbol{\eta}^H\}
= \sigma_\eta^2\mathbf{I}$ in the time domain, where $\sigma_\eta^2$
is the (sample-independent) per-sample distortion power.

Applying the unitary DAFT matrix $\mathbf{A}$,
$\mathbb{E}\{(\mathbf{A}\boldsymbol{\eta})(\mathbf{A}\boldsymbol{\eta})^H\}
= \mathbf{A}\mathbf{C}_\eta\mathbf{A}^H
= \sigma_\eta^2\mathbf{A}\mathbf{A}^H = \sigma_\eta^2\mathbf{I}$,
so the distortion covariance remains white (and hence
subcarrier-uniform) after the forward DAFT. Since $B$ is a scalar
applied identically to every time-domain sample, the DAFT-domain
Bussgang gain is likewise the same scalar $B$ on every
chirp-subcarrier.
\end{IEEEproof}

\textbf{Empirical validation.} Measured distortion power per
chirp-subcarrier is flat within the finite-sample bound (coefficient
of variation $4.7\%$ over 400 frames), and the constant-modulus
chirp pilot introduces no deterministic per-subcarrier Bussgang bias
(AM/AM compression $\approx0.99998$ for all tested $N$, IBO). The
per-subcarrier vs.\ shared-Bussgang-weight gap is $0.1$~pp
($p=0.715$), and a TOST test confirms equivalence within $1$~pp
($p_{\mathrm{TOST}}=0.0032$)---both sources of OFDM-style
per-subcarrier variability are absent here.

%% ==================================================================
\FloatBarrier
\section{Derivation of the Linewidth Threshold $\Delta f^*$}
\label{sec:supp_dfstar}
%% ==================================================================

The proposed receiver removes two structured effects: the PA's complex gain, which is a bias on the symbols (via the Bussgang scalar, Stage~1), and the CPE (via the decision-directed tracker, Stage~2). The PA's Bussgang distortion $\eta$ is uncorrelated with the signal, and no scalar can remove it.  The residual
after Stage~2 is dominated by the ICI power
$\sigma_{\mathrm{ICI}}^2(\Delta f)$. This ICI is data-dependent and
structured, and dedicated schemes can
suppress the remaining ICI further by estimating more spectral
components of $e^{j\varphi(n)}$~\cite{petrovic2007}; the scalar and
local stages of our receivers do not, so for this receiver class the
residual ICI acts as additional noise. At large $\Delta f$ the
decision-directed CPE estimate also degrades, because ICI makes the
hard decisions it relies on unreliable.

We define the heuristic threshold at which the residual ICI equals the PA distortion that the scalar stage cannot remove, $\sigma_{\mathrm{ICI}}^2 = \sigma_\eta^2$; below it the PA distortion dominates the residual error, above it the ICI term does; $\Delta f^*$
is an operating rule for this receiver class, not a fundamental limit.  Setting
$\sigma_{\mathrm{ICI}}^2 = \sigma_\eta^2$ and using the
small-$\sigma_w^2$ approximation of the ICI power in Section~\ref{sec:pn_kernel},
\begin{equation*}
\bar{P}\,\sigma_w^2\,\frac{N^2-1}{6N} = \sigma_\eta^2(\mathrm{IBO}),
\end{equation*}
where $\sigma_w^2 = 2\pi\Delta f\,T_s$ is the Wiener phase variance
per sample and $\bar{P}(\mathrm{IBO})=\mathbb{E}[g^2(\kappa|s|)\cdot|s|^2]$
is the mean PA output power.  Solving for $\Delta f$ gives the
closed-form boundary
\begin{equation}
\Delta f^*
  = \frac{6N\,\sigma_\eta^2(\mathrm{IBO})}
         {2\pi T_s\,(N^2-1)\,\bar{P}(\mathrm{IBO})}.
\label{eq:dfstar_supp}
\end{equation}

\textbf{Reproducibility: PA parameters, drive convention, and the
defining integrals.}
The Rapp AM/AM and AM/PM characteristics (\eqref{eq:rapp}) are
evaluated at $p=2$, $\alpha=0.1$~rad, $P_{\mathrm{sat}}=1$
(Section~\ref{sec:pa_model}). Write
$G(r) = g_{\mathrm{PA}}(r)\,e^{j\phi_{\mathrm{PA}}(r)}$ for the
complex Rapp response, so the PA output for drive amplitude $r$ is
$G(r)\,e^{j\angle s}$. The input back-off sets the drive scale
$\kappa = 10^{-\mathrm{IBO}/20}$: for a unit-power input symbol $s$
($\mathbb{E}|s|^2=1$), the PA is driven at $\kappa s$ and its output
is referred back to unit scale by dividing by $\kappa$,
\begin{equation*}
f_{\mathrm{PA}}(s) = \frac{1}{\kappa}\,G(\kappa\,s)
  = \frac{1}{\kappa}\,g_{\mathrm{PA}}(\kappa|s|)\,
    e^{j\phi_{\mathrm{PA}}(\kappa|s|)}\,e^{j\angle s},
\end{equation*}
which is the convention used throughout the paper and in our simulation code.
With $u \triangleq |s|^2 \sim \mathrm{Exp}(1)$ for unit-power
circular Gaussian $s$, the Bussgang gain, distortion variance, and
mean PA output power used in~\eqref{eq:dfstar_supp} are
\begin{align}
B &= \mathbb{E}_u\!\left[\tfrac{1}{\kappa}\,
      g_{\mathrm{PA}}(\kappa\sqrt{u})\,
      e^{j\phi_{\mathrm{PA}}(\kappa\sqrt{u})}\,\sqrt{u}\right],
\label{eq:B_supp}\\
\sigma_\eta^2 &= \mathbb{E}_u\!\left[\tfrac{1}{\kappa^2}\,
      g_{\mathrm{PA}}^2(\kappa\sqrt{u})\right] - |B|^2,
\label{eq:sigmaeta_supp}\\
\bar{P}(\mathrm{IBO}) &= \mathbb{E}_u\!\left[\tfrac{1}{\kappa^2}\,
      g_{\mathrm{PA}}^2(\kappa\sqrt{u})\right],
\label{eq:pbar_supp}
\end{align}
i.e.\ $\bar{P} = \sigma_\eta^2 + |B|^2$ by construction (the total
rescaled-output power splits exactly into the Bussgang-linear part
$|B|^2$ and the distortion part $\sigma_\eta^2$, since $u$ has unit
mean and $\mathbb{E}\{s^*\eta\}=0$). All three integrals can be
evaluated by direct numerical quadrature over $u\in[0,\infty)$
against the density $e^{-u}$, or by Monte Carlo over
$s\sim\mathcal{CN}(0,1)$ as in the released code; both agree to
better than $0.1\%$ at $N_{\mathrm{MC}}=4\times10^6$ samples.

\textbf{Numerical evaluation.}
With $N=64$, $T_s=0.52\,\mu\mathrm{s}$, and the Rapp PA at
four IBO values, evaluating~\eqref{eq:B_supp}--\eqref{eq:pbar_supp}
gives:

\begin{center}
\setlength{\tabcolsep}{4pt}
\begin{tabular}{cccccc}
\hline
IBO (dB) & $\kappa$ & $|B|$ & $\sigma_\eta^2$ & $\bar{P}$ & $\Delta f^*$ (Hz) \\
\hline
1 & 0.8913 & 0.7643 & 0.0225 & 0.606 & 1063 \\
3 & 0.7079 & 0.8494 & 0.0141 & 0.735 & 550  \\
5 & 0.5623 & 0.9139 & 0.0069 & 0.841 & 234  \\
8 & 0.3981 & 0.9697 & 0.0015 & 0.941 & 44   \\
\hline
\end{tabular}
\end{center}

(At $\mathrm{IBO}=3$~dB this reproduces the Bussgang constants of
Section~\ref{sec:bussgang_const},
$B = 0.8487+0.0360j$, $\arg B = 2.43^\circ$, $\sigma_\eta^2=0.0141$,
exactly, confirming the convention above is the one used throughout.)

Because Theorem~\ref{thm:kernel}(b) makes the ICI power basis-independent within the equimodular class, $\Delta f^*$ depends only on
$(N, T_s)$ and the RF chain parameters $(\sigma_\eta^2, \bar{P})$;
it is the same for AFDM, OFDM, and OCDM at equal $N$ and $T_s$,
which an earlier toy-setup experiment confirmed to within $2\%$.

\textbf{Scaling laws.}
From~\eqref{eq:dfstar_supp}: $\Delta f^* \propto 1/T_s$ (faster
symbol rate relaxes the linewidth requirement), and for fixed
sampling rate $\Delta f^* \propto 6/(N\,\bar{P}) \sim 6/N$ for
large $N$ (larger frames are more sensitive to phase noise, as
expected).  The experiments of the main text keep the subcarrier spacing fixed, where $\Delta f^*$ is nearly independent of $N$.

%% ==================================================================
\FloatBarrier
\section{Experimental Protocol and Sounding Controls}
\label{sec:supp_realistic}
%% ==================================================================

\subsection{Protocol}
\label{sec:supp_real_protocol}

\emph{Statistics.} Uncoded effects are means of per-channel relative reductions; coded effects are one minus the ratio of mean BLERs. Both use 95\% bootstrap intervals over the 30 channels. For BLER, wins are counted over \emph{informative} channels, those where the two compared receivers do not tie (for example both error-free). The $p$-values of the main text use two-sided tests from a statistics library's default (exact for $n\le50$ without ties), and the diagnostics of Appendix~\ref{sec:supp_diag} report intervals only. No multiplicity correction is applied, so $p$-values are descriptive.

Table~\ref{tab:supp_realparams} lists the configuration of the main text's experiments.

\begin{table}[!htbp]
\renewcommand{\arraystretch}{1.1}
\caption{Configuration of the Realistic-Regime and Coded Experiments}
\label{tab:supp_realparams}
\centering
\small
\begin{tabular}{@{}l>{\raggedright\arraybackslash}p{0.62\textwidth}@{}}
\hline
\textbf{Parameter} & \textbf{Value} \\
\hline
Subcarriers $N$ & 256 (also 64, 512); spacing 30~kHz \\
Waveforms & AFDM ($c_1=(2(\lceil\nu_{\max}\rceil+2)+1)/2N$, $c_2=1/2N$); OFDM ($c_1=c_2=0$) \\
CPP / CP length & 16 samples \\
Channel & 6 paths, delays $\{0,1,2,4,6,9\}$, exp.\ PDP (decay 3 samples), \\
 & Rayleigh, $\nu_i=\nu_{\max}\cos\theta_i$, $\nu_{\max}\in\{0.1,0.5,1,2,4\}$ \\
Rapp PA & $p=2$, $\alpha=0.1$~rad, IBO $\in\{3,5,8\}$~dB \\
Phase noise & Wiener, 100~Hz, $\sigma_w^2=2\pi\Delta f/f_s$ \\
Sounding & $N$ single-chirp frames, data-RMS pilots, shared PN, thermal noise \\
Modulation & Gray 16-QAM \\
Gain calibration & $4096$ known symbols ($4096/N$ frames); M2M4 on 4096 unlabeled symbols \\
Uncoded SNR & $E_b/N_0=24$~dB ($4$ bits/symbol) \\
Coding & $(133,171)_8$, $K=7$, rates $1/2$, $3/4$; 1 codeword/symbol ($4N$ bits) \\
Decoder & max-log LLRs, soft Viterbi, terminated trellis, fixed interleaver \\
Channels / seeds & 30; $13i+7$; identical noise across waveforms \\
Codewords & 64 per channel per SNR point \\
\hline
\end{tabular}
\end{table}

\subsection{Sounding Models, Gain Estimators and the HWI Baseline (Toy Setup)}
\label{sec:supp_nb50}

The order effect was first re-tested in the toy setup ($N=64$, two paths, Doppler bins $[0,2]$, IBO $=3$~dB, $24$~dB, 30 channels $\times$ 640 frames) against three objections: that independent phase noise per sounding column or low-drive pilots create the effect, that a blind gain could replace the closed form, and that the HWI baseline was mis-specified. The sounding models are \emph{indep} (independent phase noise per column, low-drive pilots, no thermal noise), \emph{clean} (no phase noise), \emph{shared} (one phase-noise realization for all columns), \emph{fd\_nn} (shared phase noise, no thermal noise, constant-modulus pilots at data RMS) and \emph{fulldrive} (fd\_nn plus thermal noise; the model used in the main text), plus a linear-PA control (\emph{PA-off}, indep sounding). Table~\ref{tab:supp_nb50ber} gives mean BERs and Table~\ref{tab:supp_nb50paired} the paired effects.

\begin{table}[!htbp]
\renewcommand{\arraystretch}{1.05}
\caption{Toy Setup, Sounding Models: Mean BER ($\times10^{-2}$), 30 Paired Channels, $24$~dB, IBO $=3$~dB}
\label{tab:supp_nb50ber}
\centering
\small
\begin{tabular}{lccccccccc}
\hline
\textbf{Sounding} & LMMSE & Tracker & \makecell{Gain\\after} & \makecell{Gain\\before} & Calibrated & \makecell{Blind\\power} & \makecell{Blind\\M2M4} & \makecell{HWI\\+ trk.} & \makecell{HWI$^{\dagger}$\\+ trk.} \\
\hline
indep & 4.20 & 3.49 & 2.08 & 1.71 & 1.66 & 1.78 & 1.69 & 2.36 & 1.68 \\
clean & 2.70 & 1.68 & 0.72 & 0.47 & 0.45 & 0.50 & 0.47 & 1.08 & 0.46 \\
shared & 4.49 & 2.93 & 1.83 & 1.17 & 0.67 & 1.18 & 1.09 & 1.81 & 1.08 \\
fd\_nn & 3.62 & 1.86 & 1.51 & 1.19 & 0.67 & 1.16 & 1.07 & 1.64 & 1.05 \\
fulldrive & 3.86 & 2.12 & 1.75 & 1.43 & 0.89 & 1.41 & 1.30 & 1.77 & 1.28 \\
PA-off (control) & 1.28 & 0.45 & 0.45 & 0.45 & 0.44 & 0.49 & 0.44 & 0.45 & 0.44 \\
\hline
\end{tabular}

\vspace{3pt}
\parbox{0.94\textwidth}{\footnotesize ``Gain before/after'': closed-form scalar before/after the tracker. Blind power: per-frame real AGC $\sqrt{\mathrm{mean}|\hat s|^2}$ before the tracker. HWI + trk.: the HWI-LMMSE of~\cite{sui2026hwi} with regularizer $\sigma_d^2\mathbf{I}$, followed by the tracker. HWI$^{\dagger}$: covariance $(\alpha^2+\sigma_d^2)\hat{\mathbf{H}}\hat{\mathbf{H}}^H+\sigma_n^2\mathbf{I}$ with per-symbol bias removal, then the tracker.}
\end{table}

\begin{table}[!htbp]
\renewcommand{\arraystretch}{1.05}
\caption{Toy Setup, Sounding Models: Paired Effects, Mean Relative BER Reduction (\%) [95\% CI] (Wins of 30)}
\label{tab:supp_nb50paired}
\centering
\small
\setlength{\tabcolsep}{3.5pt}
\begin{tabular}{lcccccc}
\hline
 & \makecell{\textbf{Gain-first vs.}\\\textbf{tracker-first}} & \makecell{\textbf{Blind AGC first}\\\textbf{vs.\ after}} & \makecell{\textbf{M2M4 vs.}\\\textbf{closed form}} & \makecell{\textbf{Calibrated vs.}\\\textbf{closed form}} & \makecell{\textbf{HWI vs.}\\\textbf{closed form}} & \makecell{\textbf{HWI$^{\dagger}$ vs.}\\\textbf{closed form}} \\
\hline
indep & \makecell{19.0 (30)\\{}[17.5, 20.5]} & \makecell{17.2 (30)\\{}[16.1, 18.2]} & \makecell{1.2 (22)\\{}[0.6, 1.8]} & \makecell{3.2 (30)\\{}[2.8, 3.5]} & \makecell{$-$27.3 (0)\\{}[$-$39.4, $-$16.2]} & \makecell{1.7 (24)\\{}[1.2, 2.3]} \\
clean & \makecell{39.2 (30)\\{}[35.8, 42.4]} & \makecell{34.0 (30)\\{}[30.8, 37.1]} & \makecell{0.3 (18)\\{}[$-$1.6, 2.1]} & \makecell{5.7 (30)\\{}[4.7, 6.9]} & \makecell{$-$78.2 (0)\\{}[$-$112.9, $-$46.6]} & \makecell{0.7 (18)\\{}[$-$1.3, 2.6]} \\
shared & \makecell{37.6 (30)\\{}[33.2, 42.4]} & \makecell{34.5 (30)\\{}[30.6, 38.6]} & \makecell{1.4 (19)\\{}[$-$9.1, 11.6]} & \makecell{30.6 (29)\\{}[21.3, 40.6]} & \makecell{$-$50.7 (5)\\{}[$-$75.9, $-$28.2]} & \makecell{1.9 (19)\\{}[$-$8.8, 12.1]} \\
fd\_nn & \makecell{23.7 (30)\\{}[19.5, 28.2]} & \makecell{18.7 (30)\\{}[16.0, 21.3]} & \makecell{2.1 (19)\\{}[$-$8.1, 11.6]} & \makecell{29.0 (29)\\{}[19.5, 39.1]} & \makecell{$-$36.5 (11)\\{}[$-$61.9, $-$14.5]} & \makecell{$-$1.7 (16)\\{}[$-$11.2, 7.9]} \\
fulldrive & \makecell{21.1 (30)\\{}[17.1, 25.5]} & \makecell{15.7 (30)\\{}[13.3, 18.0]} & \makecell{1.8 (20)\\{}[$-$7.9, 10.8]} & \makecell{26.8 (30)\\{}[17.6, 36.5]} & \makecell{$-$21.1 (11)\\{}[$-$36.6, $-$6.5]} & \makecell{$-$1.0 (16)\\{}[$-$10.1, 8.4]} \\
PA-off (control) & \makecell{0.0 (identical;\\gain $\approx1$)} & \makecell{$-$4.2 (0)\\{}[$-$5.1, $-$3.5]} & \makecell{0.4 (16)\\{}[$-$0.5, 1.4]} & \makecell{1.0 (22)\\{}[$-$2.1, 3.1]} & 0.0 (identical) & \makecell{0.8 (24)\\{}[0.5, 1.1]} \\
\hline
\end{tabular}
\end{table}

Three conclusions follow. (i)~The order effect is not a sounding artifact: gain-first beats tracker-first on 30 of 30 channels, by $19$--$39\%$ ($p\le1.7\times10^{-6}$), under every sounding model with the PA on. (ii)~It is not an LMMSE-shrinkage artifact: with a linear PA the closed-form gain is $\approx1$ so the two orders coincide, and a data-derived blind AGC applied before the tracker is $4.2\%$ \emph{worse} than after it, whereas with the PA on it is $16$--$35\%$ better. (iii)~The closed form is not unique: the blind M2M4 gain is within $0.3$--$2.1\%$ of it (CIs contain zero for the realistic soundings), a blind power AGC is $4$--$11\%$ worse, and under shared-phase-noise soundings a calibrated complex gain is $27$--$31\%$ better, consistent with it also absorbing a static phase offset of about $6^{\circ}$ left in $\hat{\mathbf{H}}_{\mathrm{eff}}$ (in a further check, the M2M4 magnitude combined with the calibrated phase reproduced the calibrated gain's BER, whereas the calibrated phase alone did not). Mean gain magnitudes under the sounding models were $|B|=0.849$ analytic, $|\hat g_{\mathrm{cal}}|=0.838$--$0.842$ (indep, clean, shared) and $0.888$ (fd\_nn, fulldrive), and M2M4 $0.845$--$0.847$ and $0.893$. As implemented with the regularizer $\sigma_d^2\mathbf{I}$, the HWI baseline loses by $21$--$78\%$ to the closed form; with the corrected covariance and bias removal it is within about $2\%$ (HWI$^{\dagger}$ column), because that receiver is itself gain-first.

\section{Additional Diagnostics}
\label{sec:supp_diag}
%% ==================================================================

These runs use the realistic regime of the main text (Table~\ref{tab:supp_realparams}): $N=256$, six paths, $\mathrm{IBO}=5$~dB, $100$~Hz phase noise, fulldrive sounding, $24$~dB, 30 channels (seeds $13i+7$), $\nu_{\max}\in\{0.1,1\}$, and identical bits, noise and phase noise across every receiver and waveform. The closed-form scalar is $|W|=1.0857$, $\arg W=-1.92^\circ$, and the scalar corrected for the pilot's own compression is $|W_p|=1.0601$, $\arg W_p=-0.55^\circ$.

\subsection{Gain Sweep}
\label{sec:supp_diag_gain}

Table~\ref{tab:supp_gainsweep} applies $mW$ to the equalizer output before the tracker. The optimum is at $m=1.00$ in all four rows. The genie scalar (the complex gain that best maps the equalizer output to the true symbols, over test frames) has $|g|/|W|=0.962$ (s.d.\ over channels $0.003$--$0.004$) and $|g|/|W_p|=0.985$, with a phase that has mean $-0.1^\circ$ and s.d.\ $5.6^\circ$ over channels. The $3.8\%$ gap to the closed form is therefore systematic, and we attribute it to the pilots' own compression; the corrected form is $1.5\%$ from the genie value, and the BER basin is flat enough that the corrected form never helped. Ideal-hardware rows in this section use $\mathrm{IBO}=60$~dB rather than a linear PA, so they can differ in the last digit from Table~\ref{tab:ber}. The phase spread is consistent with a random static offset from the shared phase-noise sounding.

\begin{table}[!htbp]
\renewcommand{\arraystretch}{1.05}
\setlength{\tabcolsep}{2pt}
\caption{Uncoded BER ($\times10^{-2}$) Versus the Gain Multiplier $m$ (Gain $=mW$, Then Tracker)}
\label{tab:supp_gainsweep}
\centering
\small
\begin{tabular}{llccccccccc}
\hline
 & $\nu_{\max}$ & 0.80 & 0.85 & 0.90 & 0.95 & 1.00 & 1.05 & 1.10 & 1.15 & 1.20 \\
\hline
AFDM & 0.1 & 3.05 & 1.61 & 0.97 & 0.69 & \textbf{0.61} & 0.64 & 0.77 & 1.00 & 1.32 \\
OFDM & 0.1 & 2.60 & 1.50 & 1.07 & 0.89 & \textbf{0.83} & 0.84 & 0.91 & 1.04 & 1.26 \\
AFDM & 1 & 4.15 & 2.42 & 1.58 & 1.18 & \textbf{1.04} & 1.08 & 1.24 & 1.52 & 1.91 \\
OFDM & 1 & 3.65 & 2.32 & 1.73 & 1.46 & \textbf{1.36} & 1.36 & 1.45 & 1.63 & 1.89 \\
\hline
\end{tabular}
\end{table}

\subsection{Mechanism}
\label{sec:supp_diag_mech}

The first-decision SER is the SER of the hard decisions before any tracking. The values in Table~\ref{tab:mech} come from this run. Win counts (gain-first lower than tracker-first): first-decision SER $19$--$20$ of 30 channels, phase error $19$--$21$ of 30, final SER $28$--$29$ of 30 (final-SER order effect: AFDM $13.1\%$ [9.5, 16.8] and $10.1\%$ [8.1, 12.2], OFDM $6.8\%$ [3.7, 9.8] and $4.8\%$ [3.4, 6.7] at $\nu_{\max}=0.1$ and $1$).

\subsection{PT-RS-Style Baseline}
\label{sec:supp_diag_ptrs}

Sixteen known unit-power QPSK symbols replace data symbols at every 16th position (positions $8,24,\dots$), a $6.25\%$ overhead; all receivers are scored on the 240 data symbols. \emph{Pilot phase}: rotate by the phase of the pilot correlation, no gain. \emph{Pilot gain}: complex least-squares scalar from the pilots. \emph{Pilot gain+trk}: pilot gain followed by the tracker. \emph{Closed gain+ph}: closed-form gain, then pilot phase. \emph{Genie gain}: per-frame complex scalar from the true data symbols (a ceiling for scalar methods, not realizable). Table~\ref{tab:supp_ptrs} gives gain-first's relative BER reduction against each alternative; positive means gain-first is better.

\begin{table}[!htbp]
\renewcommand{\arraystretch}{1.05}
\setlength{\tabcolsep}{4pt}
\caption{Gain-First (No Pilots) Versus PT-RS-Style Alternatives: Mean Relative BER Reduction (\%) [95\% CI] (Channels Won of 30; Positive = Gain-First Better)}
\label{tab:supp_ptrs}
\centering
\small
\begin{tabular}{llcccccc}
\hline
 & $\nu_{\max}$ & \makecell{vs.\ pilot\\phase} & \makecell{vs.\ closed\\gain+ph} & \makecell{vs.\ pilot\\gain} & \makecell{vs.\ pilot\\gain+trk} & \makecell{vs.\ genie\\gain} & \makecell{vs.\\tracker} \\
\hline
AFDM & 0.1 & \makecell{16.6 (27)\\{}[$-$1.9, 32.0]} & \makecell{$-$22.0 (19)\\{}[$-$55.7, 3.8]} & \makecell{$-$9.2 (23)\\{}[$-$38.2, 13.0]} & \makecell{$-$22.1 (21)\\{}[$-$56.8, 3.7]} & \makecell{$-$31.4 (14)\\{}[$-$68.4, $-$4.0]} & \makecell{29.8 (28)\\{}[24.7, 34.8]} \\
OFDM & 0.1 & \makecell{3.6 (26)\\{}[$-$14.4, 17.6]} & \makecell{$-$21.5 (21)\\{}[$-$52.6, 0.4]} & \makecell{$-$13.9 (22)\\{}[$-$42.4, 6.3]} & \makecell{$-$21.0 (20)\\{}[$-$53.3, 1.6]} & \makecell{$-$26.7 (17)\\{}[$-$61.8, $-$2.6]} & \makecell{16.1 (28)\\{}[11.1, 21.1]} \\
AFDM & 1 & \makecell{16.9 (27)\\{}[5.6, 25.8]} & \makecell{$-$12.0 (21)\\{}[$-$30.7, 2.3]} & \makecell{$-$1.3 (24)\\{}[$-$17.4, 11.1]} & \makecell{$-$9.8 (21)\\{}[$-$27.6, 3.9]} & \makecell{$-$17.2 (15)\\{}[$-$36.5, $-$2.5]} & \makecell{24.8 (29)\\{}[21.4, 28.2]} \\
OFDM & 1 & \makecell{9.1 (27)\\{}[1.8, 15.0]} & \makecell{$-$6.1 (22)\\{}[$-$15.6, 1.2]} & \makecell{$-$0.6 (23)\\{}[$-$9.4, 6.1]} & \makecell{$-$4.5 (22)\\{}[$-$13.5, 2.6]} & \makecell{$-$7.4 (19)\\{}[$-$17.0, 0.0]} & \makecell{14.5 (30)\\{}[12.0, 17.4]} \\
\hline
\end{tabular}
\end{table}

The negative entries mean gain-first has \emph{higher} BER than that alternative. The ``wins'' in this table count channels on which gain-first has lower BER (ties were not excluded here, so wins and losses need not sum to 30).

\subsection{Spreading Controls}
\label{sec:supp_diag_spread}

OCDM is the DAFT with $c_1=c_2=1/(2N)$. DFT-s-OFDM with full-band spreading is a single carrier with an ordinary cyclic prefix; its channel is sounded with constant-modulus DFT-basis pilots, because time-domain impulses would drive the PA at amplitude $\sqrt{N}$. Its lower PAPR means it is compressed less at equal IBO, and the closed-form scalar (derived for a Gaussian input) is the wrong constant for it, so its \emph{gain-first} column is penalized and the preamble column is the fair one. We computed no paired intervals for AFDM versus DFT-s-OFDM.

\begin{table}[!htbp]
\renewcommand{\arraystretch}{1.05}
\setlength{\tabcolsep}{4pt}
\caption{Spreading Controls: Mean Uncoded BER ($\times10^{-2}$) Under Real Hardware; Ideal Hardware in $\times10^{-4}$ (All Receivers Equal)}
\label{tab:supp_spread}
\centering
\footnotesize
\begin{tabular}{llcccccc}
\hline
$\nu_{\max}$ & Waveform & LMMSE & Tracker & Closed gain $\to$ trk & Trk $\to$ closed gain & Preamble gain $\to$ trk & Ideal hw. (LMMSE) \\
\hline
0.1 & AFDM & 2.067 & 0.825 & 0.611 & 0.700 & 0.499 & 0.2 \\
 & OCDM & 2.088 & 0.870 & 0.656 & 0.745 & 0.547 & 2.6 \\
 & DFT-s-OFDM & 1.701 & 0.593 & 0.567 & 0.611 & 0.404 & 3.1 \\
 & OFDM & 2.080 & 0.991 & 0.844 & 0.901 & 0.744 & 26.1 \\
\hline
1 & AFDM & 2.682 & 1.363 & 1.029 & 1.154 & 0.911 & 4.7 \\
 & OCDM & 2.759 & 1.531 & 1.229 & 1.332 & 1.118 & 28.4 \\
 & DFT-s-OFDM & 2.250 & 1.102 & 1.051 & 1.097 & 0.888 & 27.2 \\
 & OFDM & 2.759 & 1.619 & 1.379 & 1.458 & 1.269 & 60.1 \\
\hline
\end{tabular}
\end{table}

Edges over OFDM with tracker only (real hardware): $\nu_{\max}=0.1$: AFDM $21.3\%$ [13.5, 29.0], OCDM $17.9\%$ [10.2, 25.2], DFT-s-OFDM $45.2\%$ [34.7, 54.6]; $\nu_{\max}=1$: AFDM $22.0\%$ [16.9, 27.1], OCDM $7.5\%$ [3.1, 11.6], DFT-s-OFDM $36.7\%$ [30.6, 42.5]. Ideal-hardware edges over OFDM (all receivers equal): $\nu_{\max}=0.1$: AFDM $99.6\%$ [99.3, 99.9], OCDM $94.6\%$ [90.9, 97.8], DFT-s-OFDM $93.3\%$ [89.0, 96.9]; $\nu_{\max}=1$: $95.3\%$ [93.2, 97.1], $59.9\%$ [50.2, 69.2], $58.0\%$ [42.0, 69.9]. Order effects with ideal hardware are computed on BERs near $10^{-5}$--$10^{-4}$ and are not meaningful, so we do not report them. Order effects with real hardware, gain-first versus tracker-first: DFT-s-OFDM $4.7\%$ [$-$2.0, 10.3] and $4.6\%$ [2.4, 6.7].

\FloatBarrier
\end{document}